\documentclass[11pt,a4wide]{article}

\usepackage{microtype,xspace,stmaryrd}

\usepackage{amsmath,amsfonts,amssymb,epsfig,color}
\usepackage{stmaryrd}
\usepackage{amssymb}
\usepackage{amsfonts}
\usepackage{amsmath}
\usepackage{hyperref}
\usepackage{amsthm}
\usepackage{amssymb,amstext,hyperref}
\usepackage{cite}

\newtheorem{theorem}{Theorem}
\newtheorem{lemma}{Lemma}
\newtheorem{corollary}{Corollary}

\newcommand{\opt}{{\sf opt}\xspace}

\renewcommand{\a}{a}

\newcommand{\tw}{{\sf tw}}

\usepackage{cite}
\usepackage{hyperref}
\usepackage{boxedminipage}
\usepackage[english]{babel}
\usepackage[utf8]{inputenc}

\newcommand{\intv}[1]{[#1 ]}
\newcommand{\np}{{\sf NP}\xspace}
\newcommand{\p}{{\sf P}\xspace}
\newcommand{\diam}{{\sf diam}\xspace}
\newcommand{\dist}{{\sf dist}\xspace}

\newcommand{\pb}{\textsc{Diameter-Tree}\xspace}

\newcommand{\fpt}{{\sf FPT}\xspace}
\newcommand{\xp}{{\sf XP}\xspace}

\usepackage{tikz}

\newcommand{\sm}{\setminus}

\newcommand{\es}{\emptyset}
\newcommand{\Bcal}{\mathcal{B}}
\newcommand{\Ccal}{\mathcal{C}}

\newcommand{\Qcal}{\mathcal{Q}}

\newcommand{\spanning}{\mathcal{S}}

\newcommand{\Tcal}{\mathcal{T}}
\newcommand{\Rcal}{\mathcal{R}}
\newcommand{\Ocal}{\mathcal{O}}
\newcommand{\Acal}{\mathcal{A}}

\newcommand{\cost}{\texttt{cost}}

\tikzstyle{place}=[
  circle,
  fill=black,
size=2pt,
]

\newcommand{\probl}[4]
{
  \begin{flushleft}
    \fbox{
      \begin{minipage}{#4cm}
        \noindent {\textsc {#1}}\\
        {\bf Input:} #2\\
        {\bf Output:} #3
      \end{minipage}
    }
  \end{flushleft}
}

\usepackage{amsmath}
\usepackage{amssymb}

\newcommand {\set}[1]{\left\{#1\right\}}

\newtheorem{claim}{Claim}

\begin{document}

\date{}
\title{Parameterized complexity of finding a spanning tree with minimum reload cost diameter\footnote{Emails of authors: \texttt{baste@lirmm.fr}, \texttt{didem.gozupek@gtu.edu.tr}, \texttt{paul@lirmm.fr}, \texttt{sau@lirmm.fr}, \texttt{cmshalom@telhai.ac.il}, \texttt{sedthilk@thilikos.info} . Work supported by the bilateral research program of CNRS and TUBITAK under grant no.114E731, PASTA project of Université de Montpellier, TUBITAK 2221 programme, and by project DEMOGRAPH (ANR-16-CE40-0028). }}

\author{
Julien Baste\thanks{Université de Montpellier, LIRMM, Montpellier,  France.} \and
Didem Gözüpek\thanks{Department of Computer Engineering, Gebze Technical University, Kocaeli, Turkey.}\and
Christophe Paul\thanks{AlGCo project-team, CNRS, LIRMM, France.}\and
Ignasi Sau$^{\mbox{\footnotesize \S}}$\thanks{Departamento de Matemática, Universidade Federal do Ceará, Fortaleza, Brazil.} \and
Mordechai Shalom\thanks{TelHai College, Upper Galilee, 12210, Israel.}~\thanks{Department of Industrial Engineering, Boğaziçi  University, Istanbul, Turkey.} \and
Dimitrios  M. Thilikos$^{\mbox{\footnotesize \S}}$\thanks{Department of Mathematics, National and Kapodistrian University of Athens, Athens, Greece.}}

\maketitle

\begin{abstract}
\noindent We study the  minimum diameter spanning tree problem under the reload cost model (\pb for short) introduced by Wirth and Steffan~(2001). In this problem, given an undirected edge-colored graph $G$, reload costs on a path arise at a node where the path uses consecutive edges of different colors. The objective is to find a spanning tree of $G$ of minimum diameter with respect to the reload costs. We initiate a systematic study of the parameterized complexity of the \pb problem by considering the following parameters: the cost of a solution, and the treewidth and the maximum degree $\Delta$ of the input graph. We prove that \pb is {\sf para}-\np-hard for any combination of two of these three parameters, and that it is \fpt parameterized by the three of them. We also prove that the problem can be solved in polynomial time on cactus graphs. This result is somehow surprising since we prove \pb to be \np-hard on graphs of treewidth two, which is best possible as the problem can be trivially solved on forests. When the reload costs satisfy the triangle inequality, Wirth and Steffan~(2001) proved that the problem can be solved in polynomial time on graphs with $\Delta = 3$, and Galbiati~(2008) proved that it is \np-hard if $\Delta = 4$. Our results show, in particular, that without the requirement of the triangle inequality, the problem is \np-hard if $\Delta = 3$, which is also best possible. Finally, in the case where the reload costs are polynomially bounded by the size of the input graph, we prove that \pb is in \xp and {\sf W}[1]-hard parameterized by the treewidth plus~$\Delta$.
\end{abstract}

\vspace{.1cm}


\vspace{.4cm}

\section{Introduction}
\label{sec:intro}



Numerous network optimization problems can be modeled by edge-colored graphs. Wirth and Steffan introduced in \cite{WiSt01} the concept of \emph{reload cost}, which refers to the cost that arises in an edge-colored graph while traversing a vertex via two consecutive edges of different colors. The value of the reload cost depends on the colors of the traversed edges.
Although the reload cost concept has many important applications in telecommunication networks, transportation networks, and energy distribution networks, it has surprisingly received attention only recently.

In heterogeneous communication networks, routing requires switching among different technologies such as cables, fibers, and satellite links. Due to data conversion between incompatible subnetworks, this switching causes high costs, largely outweighing the cost of routing the packets within each subnetwork. The recently popular concept of vertical handover \cite{DAD09}, which allows a mobile user to have undisrupted connection during transitioning between different technologies such as 3G (third generation) and wireless local area network (WLAN), constitutes another application area of the reload cost concept. Even within the same technology, switching between different service providers incurs switching costs. Another paradigm that has received significant attention in the wireless networks research community is \emph{cognitive radio networks} (CRN), a.k.a. \emph{dynamic spectrum access networks}. Unlike traditional wireless technologies, CRNs operate across a wide frequency range in the spectrum and frequently requires frequency switching; therefore, the frequency switching cost is indispensable and of paramount importance. Many works in the CRNs literature focused on this frequency switching cost from an application point of view (for instance, see \cite{AAK+13, GBA13, BaAl13, BEAT13, EBT13, AgDe16, ShRa16}) by analyzing its various aspects such as delay and energy consumption. Operating in a wide range of frequencies is indeed a property of not only CRNs but also other 5G technologies. Hence, applications of the reload cost concept in communication networks continuously increment. In particular, the energy consumption aspect of this switching cost is especially important in the recently active research area of green networks, which aim to tackle the increasing energy consumption of information and communication technologies \cite{BCRR12, CeKa16}.

Reload cost concept finds applications also in other networks such as transportation networks and energy distribution networks. For instance, a cargo transportation network uses different means of transportation. The loading and unloading of cargo at junction points is costly and this cost may even outweigh the cost of carrying the cargo from one point to another \cite{Ga08}. In energy distribution networks, reload costs can model the energy losses that occur at the interfaces while transferring energy from one type of carrier to another \cite{Ga08}.

Recent works in the literature focused on numerous problems related to the reload cost concept: the minimum reload cost cycle cover problem \cite{GGM14}, the problems of finding a path, trail or walk with minimum total reload cost between two given vertices \cite{GLMM10}, the problem of finding a spanning tree that minimizes the sum of reload costs of all paths between all pairs of vertices \cite{GGR12}, various path, tour, and flow problems related to reload costs \cite{AGM11}, the minimum changeover cost arborescence problem \cite{GGM11, GozupekSVZ14, FPT-by-tw-Delta, GOP+16}, and problems related to finding a proper edge coloring of the graph so that the total reload cost is minimized \cite{GoSh16}.

The work in \cite{WiSt01}, which introduced the concept of reload cost, focused on the following problem, called {\sc Minimum Reload Cost Diameter Spanning Tree} (\pb for short), and which is the one we study in this paper: given an undirected graph $G = (V,E)$ with an edge-coloring $\chi: E(G)\rightarrow X$ and a reload cost function $c:X^2 \rightarrow  \mathbb{N}_{0}$, find a spanning tree of $G$ with minimum diameter with respect to the reload costs (see Section~\ref{sec:prelim} for the formal definitions).

This problem has important applications in communication networks, since forming a spanning tree is crucial for broadcasting control traffic such as route update messages. For instance, in a multi-hop cognitive radio network where a frequency is assigned to each wireless link depending on availabilities of spectrum bands, delay-aware broadcasting of control traffic necessitates the forming of a spanning tree by taking the delay arising from frequency switching at every node into account. Cognitive nodes send various control information messages to each other over this spanning tree. A spanning tree with minimum reload cost diameter in this setting corresponds to a spanning tree in which the maximum frequency switching delay between any two nodes on the tree is minimized. Since control information is crucial and needs to be sent to all other nodes in a timely manner, ensuring that the maximum delay is minimum is vital in a cognitive radio network.

Wirth and Steffan~\cite{WiSt01} proved that \pb is inapproximable within a factor better than $3$ (in particular, it is {\sf NP}-hard), even on graphs with maximum degree $5$. They also provided a polynomial-time exact algorithm for the special case where the maximum degree is $3$ and the reload costs satisfy the triangle inequality. Galbiati~\cite{Ga08} showed stronger hardness results for this problem, by proving that even on graphs with maximum degree $4$, the problem cannot be approximated within a factor better than $2$ if the reload costs do not satisfy the triangle inequality, and cannot be approximated within any factor better than $5/3$ if the reload costs satisfy the triangle inequality. The complexity of \pb  (in the general case) on graphs with maximum degree $3$ was left open.

\vspace{.25cm}
\noindent \textbf{Our results}. In this article we initiate a systematic study of the complexity of the \pb problem, with special emphasis on its parameterized complexity for several choices of the parameters. Namely, we consider any combinations of the parameters $k$ (the cost of a solution), $\tw$ (the treewidth of the input graph), and $\Delta$ (the maximum degree of the input graph). We would like to note that these parameters have practical importance in communication networks. Indeed, besides the natural parameter $k$, whose relevance is clear, many networks that model real-life situations appear to have small treewidth~\cite{Jen01,SpLa88}. On the other hand,  the degree of a node in a network is related to its number of transceivers, which are costly devices in many different types of networks such as optical networks \cite{KoCh01}. For this reason, in practice the maximum degree of a network usually takes small values.

Before elaborating on our results, a summary of them can be found in Table~\ref{tab:results}.

\begin{table}[htb]
\begin{center}
\scalebox{.913}{\begin{tabular}{|c||c|c|c|c||c|}
\hline
Problem  & \multicolumn{4}{c||}{Parameterized complexity with parameter} & Polynomial \\

\cline{2-5}
  & $k + \tw$ & $k + \Delta$ &  $\tw + \Delta$ & $k + \tw + \Delta$ & cases \\
 \hline  \hline

    &    {\sf NP}h for     &  {\sf NP}h for   &   {\sf NP}h for   &   {\sf FPT}   &       in \p on          \\

 \pb        &   $k=9, \tw=2$     &   $k=0, \Delta = 3$  &  $\tw=3, \Delta=3$  &     (Thm~\ref{thm:FPT-algo})    &  cacti       \\
           &   (Thm~\ref{thm:hard-bounded-tw})     &  (Thm~\ref{thm:hard-bounded-degree}) & (Thm~\ref{thm:hard-tw-delta})    &   &        (Thm~\ref{thm:cactus})         \\

 \hline
\pb                    &       &     &   \xp  (Thm~\ref{thm:FPT-algo})&       &                    \\
     with poly costs    &    $\checkmark$   &    $\checkmark$  &   {\sf W}[1]-hard    &  $\checkmark$   &     $\checkmark$       \\
                      &       &    &   (Thm~\ref{thm:W-hard-poly-costs})  &      &                     \\
\hline
\end{tabular}
}
\end{center}\vspace{-.25cm}
\caption{\label{tab:results} Summary of our results, where $k, \tw, \Delta$ denote the cost of the solution,  the treewidth, and the maximum degree of the input graph, respectively. {\sf NP}h   stands for {\sf NP}-hard. The symbol `$\checkmark$' denotes that the result above still holds for polynomial costs.\vspace{-.15cm}}
\end{table}

We first prove, by a reduction from \textsc{3-Sat}, that \pb  is \np-hard on outerplanar graphs (which have treewidth at most 2) with only one vertex of degree greater than 3, even with three different costs that satisfy the triangle inequality, and $k=9$. Note that, in the case where the costs satisfy the triangle inequality, having only one vertex of degree greater than 3 is best possible, as if all vertices have degree at most 3, the problem can be solved in polynomial time~\cite{WiSt01}. Note also that the bound on the treewidth is best possible as well, since the problem is trivially solvable on graphs of treewidth 1, i.e., on forests.

Toward investigating the border of tractability of the problem with respect to treewidth, we exhibit a polynomial-time algorithm on a relevant subclass of the graphs of treewidth at most 2: cactus graphs. This algorithm is quite involved and, in a nutshell, processes in a bottom-up manner the \emph{block tree} of the given cactus graph, and  uses at each step of the processing an algorithm that solves \textsc{2-Sat} as a subroutine.

Back to hardness results, we also prove,  by a reduction from a restricted version of \textsc{3-Sat}, that \pb  is $\np$-hard on graphs with $\Delta \leq 3$, even with only two different costs, $k=0$, and bounded number of colors. In particular, this settles the complexity of the problem on graphs with $\Delta \leq 3$ in the general case where the triangle inequality is not necessarily satisfied, which had been left open in previous work~\cite{WiSt01,Ga08}. Note that $\Delta \leq 3$ is best possible, as \pb can be easily solved on graphs with $\Delta \leq 2$.

As our last $\np$-hardness reduction, we prove, by a reduction from \textsc{Partition}, that
the \pb problem is $\np$-hard on planar graphs with $\tw \leq 3$ and $\Delta \leq 3$.

The above hardness results imply that  the \pb problem is {\sf para}-\np-hard for any combination of two of the three parameters $k$, $\tw$, and $\Delta$. On the positive side, we show that \pb is \fpt parameterized by the three of them, by using a (highly nontrivial) dynamic programming algorithm on a tree decomposition of the input graph.

 Since our {\sf para}-$\np$-hardness reduction with parameter $\tw + \Delta$ is from \textsc{Partition}, which is a typical example of {\sl weakly} $\np$-complete problem~\cite{GJ79}, a natural question is whether \pb, with parameter $\tw + \Delta$,  is {\sf para}-\np-hard, \xp, {\sf W}[1]-hard, or \fpt  when the reload costs are {\sl polynomially bounded} by the size of the input graph. We manage to answer this question completely: we show that in this case the problem is in \xp (hence {\sl not} {\sf para}-$\np$-hard) and {\sf W}[1]-hard parameterized by $\tw + \Delta$. The {\sf W}[1]-hardness reduction is from the
\textsc{Unary Bin Packing} problem parameterized by the number of bins, proved to be ${\sf W}[1]$-hard by Jansen \emph{et al}.~\cite{JansenKMS13}.

Altogether, our results provide an accurate picture of the (parameterized) complexity of the \pb problem.

%
%
%
%
%

\vspace{.25cm}
\noindent \textbf{Organization of the paper}. We start in Section~\ref{sec:prelim} with some brief preliminaries about graphs, the \pb problem, parameterized complexity, and tree decompositions. In Section~\ref{sec:NPhardness} we provide the {\sf para}-\np-hardness results. In Section~\ref{sec:cactus} we present  the polynomial-time algorithm on cactus graphs, and in Section~\ref{sec:FTP-algo} we present the  \fpt algorithm on general graphs parameterized by $k + \tw + \Delta$. In Section~\ref{sec:poly-costs} we focus on the case where the reload costs are polynomially bounded. Finally, we conclude the article in Section~\ref{sec:conclusion}.

\section{Preliminaries}
\label{sec:prelim}
\noindent \textbf{Graphs and sets}. We use standard graph-theoretic notation, and we refer the reader to~\cite{Die05} for any undefined term. Given a graph $G$ and a set $S \subseteq V(G)$, we define $\textsf{adj}_G(S)$ to be the set of edges of $G$ that intersect $S$.
We also define $N_G[S] = S \cup \{x \mid \exists y \in S : \{x,y\} \in E(G)\}$. For a graph $G$ and an edge $e\in E(G)$, we let $G - e = (V(G), E(G) \setminus e)$. Given a graph $G$ and a set $S\subseteq V(G)$, we say that $S$  is {\em good for}  $G$ if
each connected component of $G$ contains at least one vertex of $S$. Given two integers $i$ and $j$ with $i \leq j$, we denote by $\intv{i,j}$  the set of all integers $k$ such that $i \leq k \leq j$. For an integer $i \geq 1$, we denote by $[i]$  the set of all integers $k$ such that $ 1 \leq k \leq i$.

\vspace{.25cm}
\noindent \textbf{Reload costs and definition of the problem.} For reload costs, we follow the notation and terminology defined by \cite{WiSt01}. We consider edge-colored graphs $G=(V,E)$, where the colors are taken from a finite set $X$ and the coloring function is $\chi: E(G)\rightarrow X$. The reload costs are given by a nonnegative function $c:X^2 \rightarrow  \mathbb{N}_{0}$, which we assume for simplicity to be symmetric. The cost of traversing two incident edges $e_1$, $e_2$ is $c(e_1, e_2):=c(\chi(e_1), \chi(e_2))$. By definition, reload costs at the endpoints of a path equal zero. Consequently, the reload cost of a path with one edge also equals zero. The \emph{reload cost} of a path $P$ of length $\ell \geq 2$ with edges $e_1,e_2,\ldots,e_{\ell}$ is defined as $c(P):=\sum_{i=2}^{\ell}c(e_{i-1}, e_{i})$. The induced reload cost distance function is given by $\dist_{G}^{c}(u,v)=\min\{c(P) \mid \text{$P$ is a path from $u$ to $v$ in $G$}\}$. The \emph{diameter} of a tree $T$ is $\diam(T):=\max_{u,v\in V} \dist_{T}^{c}(u,v)$, where for notational convenience we assume that the edge-coloring function $\chi$ and the reload cost function $c$ are clear from the context.

The problem we study in this paper can be formally defined as follows:

\probl
{\sc Minimum Reload Cost Diameter Spanning Tree (\pb)}
{A graph $G = (V,E)$ with an edge-coloring $\chi$ and a reload cost function $c$.}
{A spanning tree $T$ of $G$ minimizing $\diam(T)$.}{14.1}

If for every three distinct edges $e_1,e_2,e_3$ of $G$ incident to the same node, it holds that $c(e_1,e_3) \leq c(e_1,e_2) + c(e_2,e_3)$, we say that the reload cost function $c$ satisfies the \emph{triangle inequality}. This assumption is sometimes used in practical applications~\cite{WiSt01}.

%


\vspace{.25cm}
\noindent \textbf{Parameterized complexity.} We refer the reader to~\cite{DF13,CyganFKLMPPS15} for basic background on parameterized complexity, and we recall here only some basic definitions.
A \emph{parameterized problem} is a language $L \subseteq \Sigma^* \times \mathbb{N}$.  For an instance $I=(x,k) \in \Sigma^* \times \mathbb{N}$, $k$ is called the \emph{parameter}. 

A parameterized problem is \emph{fixed-parameter tractable} ({\sf FPT}) if there exists an algorithm $\Acal$, a computable function $f$, and a constant $c$ such that given an instance $I=(x,k)$,
$\Acal$ (called an {\sf FPT} \emph{algorithm}) correctly decides whether $I \in L$ in time bounded by $f(k) \cdot |I|^c$. For instance, the \textsc{Vertex Cover} problem parameterized by the size of the solution is \fpt.

A parameterized problem is in \xp if there exists an algorithm $\Acal$ and two computable functions $f$ and $g$ such that given an instance $I=(x,k)$,
$\Acal$ (called an \xp \emph{algorithm}) correctly decides whether $I \in L$ in time bounded by $f(k) \cdot |I|^{g(k)}$. For instance,  the \textsc{Clique} problem parameterized by the size of the solution is in \xp.

A parameterized problem with instances of the form $I=(x,k)$ is \emph{{\sf para}-\np-hard} if it is \np-hard for some fixed {\sl constant} value of the parameter $k$. For instance, the \textsc{Vertex Coloring} problem parameterized by the number of colors is {\sf para}-\np-hard. Note that, unless $\p = \np$, a {\sf para}-\np-hard problem cannot be in \xp, hence it cannot be \fpt either.

Within parameterized problems, the class {\sf W}[1] may be seen as the parameterized equivalent to the class \np of classical optimization problems. Without entering into details (see~\cite{DF13,CyganFKLMPPS15} for the formal definitions), a parameterized problem being {\sf W}[1]-\emph{hard} can be seen as a strong evidence that this problem is {\sl not} \fpt. For instance,  the \textsc{Clique} problem parameterized by the size of the solution is the canonical example of a {\sf W}[1]-hard problem. To transfer {\sf W}[1]-hardness from one problem to another, one uses a \emph{parameterized reduction}, which given an input $I=(x,k)$ of the source problem, computes in time $f(k) \cdot |I|^c$, for some computable function $f$ and a function $c$, an equivalent instance $I'=(x',k')$ of the target problem, such that $k'$ is bounded by a function depending only on $k$.

\vspace{.25cm}
\noindent \textbf{Tree decompositions.} A \emph{tree decomposition} of a graph $G$ is a pair ${\cal D}=(Y,{\cal X})$, where $Y$ is a tree
and ${\cal X}=\{X_{t}\mid t\in V(Y)\}$ is a collection of subsets of $V(G)$
 such that:
\begin{itemize}
\item $\bigcup_{t \in V(Y)} X_t = V(G)$,
\item for every edge $\{u,v\} \in E$, there is a $t \in V(Y)$ such that $\{u, v\} \subseteq X_t$, and
\item for each $\{x,y,z\} \subseteq V(Y)$ such that $z$ lies on the unique path between $x$ and $y$ in $Y$,  $X_x \cap X_y \subseteq X_z$.
\end{itemize}
We call the vertices of $Y$ {\em nodes} of ${\cal D}$ and the sets in ${\cal X}$ {\em bags} of ${\cal D}$. The width of the  tree decomposition ${\cal D}=(Y,{\cal X})$ is $\max_{t \in V(Y)} |X_t| - 1$.
The \emph{treewidth} of $G$, denoted by $\tw(G)$, is the smallest integer $w$ such that there exists a tree decomposition of $G$ of width at most $w$.

\vspace{.25cm}
\noindent \textbf{Nice tree decompositions.} Let ${\cal D}=(Y,{\cal X})$
be a tree decomposition of $G$, $r$ be a vertex of $Y$, and   ${\cal G}=\{G_{t}\mid t\in V(Y)\}$ be
a collection of subgraphs of   $G$, indexed by the vertices of $Y$.
We say that the triple $({\cal D},r,{\cal G})$ is
\emph{nice} if the following conditions hold:
\begin{itemize}

\item $X_{r} = \es$ and $G_{r}=G$,
\item each node of ${\cal D}$ has at most two children in $Y$,
\item for each leaf $t \in V(Y)$, $X_t = \es$ and $G_{t}=(\emptyset,\emptyset).$ Such a $t$ is called a {\em leaf node},
\item if $t \in V(\Tcal)$ has exactly one child $t'$, then either
\begin{itemize}
\item $X_t = X_{t'}\cup \{v_{\rm insert}\}$ for some $v_{\rm insert} \not \in X_{t'}$ and $G_{t}=(V(G_{t'})\cup\{v_{\rm insert}\},E(G_{t'}))$.
  The node $t$ is called \emph{vertex-introduce  node}   and the vertex $v_{\rm insert}$ is the {\em insertion vertex} of $X_{t}$,
\item $X_t = X_{t'}$ and $G_{t}=(G_{t'},E(G_{t'})\cup\{e_{\rm insert}\})$ where $e_{\rm insert}$ is an edge of $G$ with endpoints in  $X_{t}$.
  The node $t$ is called \emph{edge-introduce  node}  and the edge $e_{\rm insert}$ is the {\em insertion edge} of $X_{t}$, or
\item $X_t = X_{t'} \sm \{v_{\rm forget}\}$ for some $v_{\rm forget} \in X_{t'}$ and $G_{t}=G_{t'}$.
  The node $t$ is called   \emph{forget node} and $v_{\rm forget}$ is the {\em forget vertex} of $X_{t}$.
\end{itemize}
\item if $t \in V(Y)$ has exactly two children $t'$ and $t''$, then $X_{t} = X_{t'} = X_{t''}$, and $E(G_{t'})\cap E(G_{t''})=\emptyset$. The node $t$ is called a \emph{join node}.
\end{itemize}
 The notion of a nice triple defined above is essentially the same
as the one of nice tree decomposition in~\cite{CyganNPPRW11} (which in turn is an enhancement of the original one, introduced in~\cite{Klo94}).
As already argued in~\cite{CyganNPPRW11,Klo94}, it is possible, given a tree decomposition to transform it in polynomial time
to a new one ${\cal D}$ of the same width and construct a collection ${\cal G}$ such that the triple $({\cal D},r,{\cal G})$ is nice.

\vspace{.25cm}
\noindent \textbf{Transfer triples and their fusion.} Let $(F,R,\alpha)$ be a triple where $F$ is a forest,  $R\subseteq V(F)$, and
$\alpha: R\times R^{F}\to\intv{0,k}\cup\{\bot\}$, where $R^{F}=V(F)\cup (E(F)\setminus{\sf adj}_{F}(R))$.
Keep in mind that $R^{F}$ contains all vertices and edges of $F$ except from the edges that are incident to vertices in $R$.
We call $(F,R,\alpha)$
a {\em transfer triple}  if, given a $(v,a)\in R\times  R^{F}$, $\alpha(v,a)=\bot$ if and only if $v$ and $a$ belong in different connected components of $F$. The function $\alpha$ will be used for indicating for each pair $(v,a)$  the ``cost of transfering'' from $v$ to $a$ in $F$ ($\alpha$ is not necessarily a distance function).

Let $(F_{1},R_{1},\alpha_{1})$ and  $(F_{2},R_{2},\alpha_{2})$ be two transfer triples where $R=R_{1}=R_{2}$, $E(F_{1})\cap E(F_{2})=\emptyset$,
and such that  $F=F_{1}\cup F_{2}$ is a forest.
Let also $\beta: {\sf adj}_{F_{1}}(R)\times{\sf adj}_{F_{2}}(R)\to\intv{0,k}\cup\{\bot\}$.
We require a function $\alpha_{1}{\oplus}_{\beta}\alpha_{2}: R\times R^{F}\to\intv{0,k}\cup\{\bot\}$
that builds the transferring costs of moving in $F$ by taking into account the corresponding transferring costs in $F_{1}$ and $F_{2}$.
The values of $\alpha_{1}{\oplus}_{\beta}\alpha_{2}$ are defined as follows:

Let $(v,a)\in R\times  R^{F}$.
Let $P$ be the  shortest path in $F$ containing $v$ and $a$ and let $V(P)=\{v_{0},\ldots,v_{r}\}$,
ordered in the way these vertices appear in $P$ and assuming that $v_0=v$.
To simplify notation, we assume that $\{v_{0},v_{1}\}$ is an edge of $F_{1}$
(otherwise, exchange the roles of $F_{1}$ and $F_{2}$). Given  $i\in\intv{r-1}$,
we define $e^{-}_{i}$ (resp. $e^{+}_{i}$) as the edge incident to $v_{i}$ that
appears before (resp. after) $v_{i}$ when traversing $P$ from $v$ to $a$.
We define the set of indices $$I=\{i\mid \mbox{$e^-_{i}$ and $e^+_{i}$ belong to different sets of $\{E(F_{1}),E(F_{2})\}$}\}.$$
Let $I=\{i_{1},\ldots,i_{q}\}$, where numbers are ordered in increasing order and we also set $i_0=0$. Then we set
\begin{eqnarray*}
\alpha_{1}\oplus_{\beta}\alpha_{2}(v,a)
& = & \sum_{h\in\intv{0,\lfloor \frac{q-1}{2}\rfloor}}\alpha_{1}(v_{2i_{h}},v_{2i_{h}+1}) +\sum_{h\in\intv{0,\lfloor \frac{q-2}{2}\rfloor}}\alpha_{2}(v_{2i_{h}+1},v_{2i_{h}+2})\\
& & +\sum_{h\in\intv{q}}\beta(e_{i_{h}}^-,e_{i_{h}}^+) +\alpha_{(q\!\!\!\!\mod 2)+1}(v_{i_{q}},a).
\end{eqnarray*}

\medskip

Throughout the paper, we let $n$, $\Delta$, and $\tw$ denote the number of vertices, the maximum degree, and the treewidth of the input graph, respectively. When we consider the (parameterized) decision version of the \pb problem, we let also $k$ denote the desired cost of a solution.

\section{{\sf Para}-NP-hardness results}\label{sec:NPhardness}

We start with the {\sf para}-$\np$-hardness result with parameter $k + \tw$.

\begin{theorem}\label{thm:hard-bounded-tw}
The \pb problem is $\np$-hard on outerplanar graphs with only one vertex of degree greater than 3, even with three different costs that satisfy the triangle inequality, and $k=9$. Since outerplanar graphs have treewidth at most 2, in particular, \pb is {\sf para}-$\np$-hard parameterized by $\tw$ and $k$.
\end{theorem}
\begin{proof} We present a simple reduction from \textsc{3-Sat}. Given a formula $\varphi$ with $n$ variables and $m$ clauses, we create an instance $(G, \chi, c)$ of \pb as follows.
We may assume that there is no clause in $\varphi$ that contains a literal and its negation.
The graph $G$ contains a distinguished vertex $r$ and, for each clause $c_j = (\ell_1 \vee \ell_2 \vee \ell_3)$, we add a clause gadget $C_j$ consisting of three vertices $v^j_{\ell_1}, v^j_{\ell_2}, v^j_{\ell_3}$ and five edges $\{r, v^j_{\ell_1}\}$, $\{r, v^j_{\ell_2}\}$, $\{r, v^j_{\ell_3}\}$, $\{v^j_{\ell_1}, v^j_{\ell_2}\}$, and $\{v^j_{\ell_2}, v^j_{\ell_3}\}$. This completes the construction of $G$. Note that $G$ does not depend on the formula $\varphi$ except for the number of clause gadgets, and that it is an outerplanar graph with only one vertex of degree greater than 3; see Figure~\ref{fig:Occitan-graph} for an illustration.

\begin{figure}[h!]
\begin{center}
\includegraphics[width=0.33\textwidth]{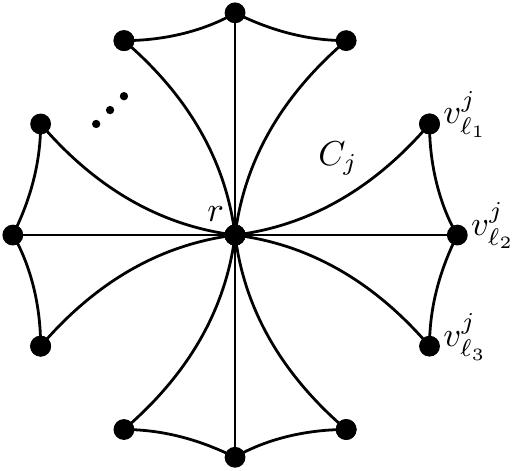}
\end{center}
\caption{Example of the graph $G$ built in the reduction of Theorem~\ref{thm:hard-bounded-tw}.}
\label{fig:Occitan-graph}
\end{figure}

Let us now define the coloring $\chi$ and the cost function $c$. For simplicity, we associate a distinct color with
each edge of $G$, and thus, with slight abuse of notation, it is enough to describe the
cost function $c$ for every pair of incident edges of $G$, as we consider symmetric cost functions.
We will use just three different costs: 1, 5 and 10.
We set
\[
c(e_1,e_2) = \left\{
\begin{array}{ll}
10 & \textrm{if~} e_1 = \{r, v^{j_1}_{\ell_{i_1}}\}, e_2 = \{r, v^{j_2}_{\ell_{i_2}}\}\textrm{~and~} \ell_{i_1} = \overline{\ell_{i_2}} \textrm{~,~} \\
5 & \textrm{if~} e_1 = \{r, v^{j_1}_{\ell_{i_1}}\}, e_2 = \{r, v^{j_2}_{\ell_{i_2}}\}\textrm{~and~} \ell_{i_1} \not = \overline{\ell_{i_2}} \textrm{~, and~} \\
1 & \textrm{otherwise}.
\end{array}
\right.
\]
Note that this cost function satisfies the triangle inequality since the reload costs between edges incident to $r$ are 5 and 10, and the reload costs between edges incident to other vertices are 1.

We claim that $\varphi$ is satisfiable if and only if $G$ contains a spanning tree with diameter at most $9$.
Since $r$ is a cut vertex and every clause gadget is a connected component of $G - r$, in every spanning tree, the vertices of $C_j$ together with $r$ induce a tree with four vertices.
Moreover  the reload cost associated with a path from $r$ to a leaf of this tree is always at most $2$.
Therefore, the diameter of any spanning tree is at most $4$ plus the maximum reload cost incurred at $r$ by a path of $T$.

Assume first that $\varphi$ is satisfiable, fix a satisfying assignment $\psi$ of $\varphi$, and let us construct a spanning tree $T$ of $G$ with diameter at most $9$. For each clause $c_j$, the tree $T^j$ is the tree spanning $C_j$ and containing the edge between $r$ and an arbitrarily chosen literal of $c_j$ that is set to true by $\psi$. $T$ is the union of all the trees $T_j$ constructed in this way. The reload cost incurred at $r$ by any path of $T$ traversing it is at most $5$, since we never choose a literal and its negation. Therefore, it holds that $\diam(T)\leq 9$.

Conversely, let $T$ be a spanning tree of $G$ with $\diam(T)\leq 9$. Then, the reload cost incurred at $r$ by any path traversing it is at most $5$ since otherwise $\diam(T) \geq 10$. For every $j \in [m]$, let $T_j$ be the subtree of $T$ induced by $C_j$ and let $\{r, v^j_{\ell_{i_j}} \}$ be one of the edges incident to $r$ in $T_j$. We note that for any pair of clauses $c_{j_1}, c_{j_2}$ we have $\ell_{i_{j_1}} \neq \overline{ \ell_{i_{j_2}}}$, since otherwise a path using these two edges would incur a cost of $10$ at $r$. The variable in the literal $\ell_{i_j}$ is set by $\psi$ so that $\ell_{i_j}$ is true. All the other variables are set to an arbitrary value by $\psi$. Note that $\psi$ is well-defined, since we never encounter a literal and its negation during the assignment process. It follows that $\psi$ is a satisfying assignment of $\varphi$.
\end{proof}


We proceed with the {\sf para}-$\np$-hardness result with parameter $k + \Delta$.


\begin{theorem}\label{thm:hard-bounded-degree}
The \pb problem is $\np$-hard on graphs with $\Delta \leq 3$, even with two different costs, $k=0$, and bounded number of colors. In particular, it is {\sf para}-$\np$-hard parameterized by $k$ and $\Delta$.
\end{theorem}
\begin{proof}
We present a reduction from the restriction of \textsc{3-Sat} to formulas where each variable occurs in at most three clauses; this problem was proved to be $\np$-complete by Tovey~\cite{Tov84}. It is worth mentioning that one needs to allow for clauses of size two or three, as if all clauses have size exactly three, then it turns out that all instances are satisfiable~\cite{Tov84}.

We may assume that each variable occurs at least once positively and at least once negatively, as otherwise we may set such variable $x$ to the value that satisfies all clauses in which it appears, and delete $x$ together with those clauses from the formula. We may also assume that each variable occurs {\sl exactly} three times in the given formula $\varphi$. Indeed, let $x$ be a variable occurring exactly two times in the formula. We create a new variable $y$ and we add to $\varphi$ two clauses $(x \vee y)$ and $(y \vee \overline{y})$. Let $\varphi'$ be the new formula. Clearly $\varphi$ and $\varphi'$ are equivalent, and both $x$ and $y$ occur three times in $\varphi'$. Applying these operations exhaustively clearly results in an equivalent formula where each variable occurs exactly three times. Summarizing, we may assume the following property:

\begin{itemize}
\item[$\maltese$] \emph{Each variable occurs exactly three times in the given formula $\varphi$ of \textsc{3-Sat}. Moreover, each variable occurs at least once positively and at least once negatively in $\varphi$.}
\end{itemize}

\noindent Given a formula $\varphi$ with $n$ variables and $m$ clauses, we create an instance $(G, \chi, c)$ of \pb with $\Delta(G)\leq 3$ as follows. Let the variables in $\varphi$ be $x_1, \ldots, x_n$.  For every  $i \in [n]$, we add to $G$ a \emph{variable gadget} consisting of five vertices $u_i, v_i, p_i,r_i,n_i$ and five edges $\{u_i, v_i\}, \{v_i, p_i\}, \{p_i, r_i\}, \{r_i, n_i\}$, and $\{n_i, v_i\}$. For every $ i  \in [n-1]$, we add the edge $\{u_i, u_{i+1}\}$. For every $ j \in [m]$, the clause gadget in $G$ consists of a single vertex $c_j$. We now proceed to explain how we connect the variable and the clause gadgets. For each variable $x_i$, we connect vertex $p_i$ (resp. $n_i$) to one of the vertices corresponding to a clause of $\varphi$ in which $x_i$ appears positively (resp. negatively). Finally, we connect vertex $r_i$ to the remaining clause in which $x_i$ appears (positively or negatively). Note that these connections are well-defined because of property~$\maltese$. This completes the construction of $G$, and note that it indeed holds that $\Delta(G)\leq 3$; see Figure~\ref{fig:hard3}(a) for an example of the construction of $G$ for a specific satisfiable formula $\varphi$ with $n=4$ and $m=5$.

Let us now define the coloring $\chi$ and the cost function $c$. We use nine colors $1,2,\ldots,9$ associated with the edges of $G$ as follows. For $i \in [n]$, we set $\chi(\{p_i,r_i\}) = 1$ and $\chi(\{r_i,n_i\}) = 2$, and all edges incident to $u_i$ or $v_i$ have color 3. Finally, for $j \in [m]$, we color the edges containing $c_j$ with colors in $\{4,5,6,7,8,9\}$, so that incident edges get different colors, and edges corresponding to positive (resp. negative) occurrences get colors in $\{4,5,6\}$ (resp. $\{7,8,9\}$); note that such a coloring always exists as each clause contains at most three variables; see Figure~\ref{fig:hard3}(b). We will use only two costs, namely 0 and 1, and recall that we consider just symmetric cost functions. We set $c(1,2)=1$, $c(1,i) = 1$  for every $i \in \{4,5,6\}$, $c(2,i) = 1$  for every $i \in \{7,8,9\}$, and $c(i,j) = 1$  for every distinct $4 \leq i,j \leq 9$. All other costs are set to 0.

\begin{figure}[htb]
\begin{center}
\includegraphics[width=.95\textwidth]{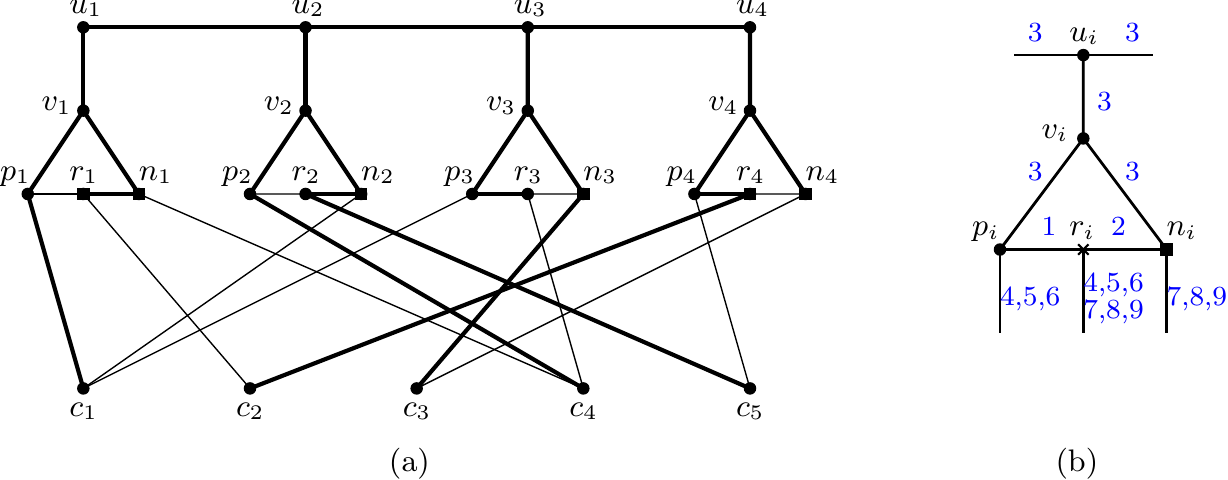}
\end{center}
\caption{(a) Graph $G$ described in the reduction of Theorem~\ref{thm:hard-bounded-degree} for the formula $\varphi = (x_1 \vee \overline{x_2} \vee x_3) \wedge (\overline{x_1}  \vee \overline{x_4}) \wedge (\overline{x_3} \vee \overline{x_4}) \wedge (\overline{x_1} \vee x_2 \vee x_3) \wedge (x_2 \vee x_4)$. The vertices $p_i,r_i,n_i$ corresponding to positive (resp. negative) occurrences are depicted with circles (resp. squares). An assignment satisfying $\varphi$ is given by $x_1 = x_2 = 1$ and $x_3 = x_4 = 0$,  and a solution spanning tree $T$ with diameter 0 is emphasized with thicker edges. (b) The (possible) colors associated with edge edge of $G$ are depicted in blue.}
\label{fig:hard3}
\end{figure}

We claim that $\varphi$ is satisfiable if and only if $G$ contains a spanning tree with diameter $0$. Assume first that $\varphi$ is satisfiable, fix a satisfying assignment $\psi$ of $\varphi$, and let us construct a spanning tree $T$ of $G$ with diameter $0$. For every $ i \in [n]$, tree $T$ contains all the edges containing vertex $u_i$ or $v_i$. If variable $x_i$ is set to true by  $\psi$, we include the edge $\{r_i,n_i\}$ to $T$, and otherwise, that is, if  $x_i$ is set to false by  $\psi$, we include the edge $\{p_i, r_i\}$. Finally, for $j \in [m]$, we add to $T$ one of the edges containing $c_j$ that corresponds to a literal satisfying that clause. It can be easily checked that $T$ is a spanning tree of $G$ with diameter 0; see Figure~\ref{fig:hard3}(a) for an example.

Conversely, let $T$ be a spanning tree of $G$ with diameter $0$. Since the cost associated with any two distinct colors in $\{4,5,6,7,8,9\}$ is 1, it follows that, for $j \in [m]$, vertex $c_j$ has degree one in $T$. Therefore, the variable gadgets need to be connected in $T$ via the vertices $u_i$, implying that all edges containing $u_i$, for $i \in [n]$, belong to $T$. For $i \in [n]$, in order for $T$ to contain all four vertices $v_i,p_i,r_i,n_i$, by construction of $G$ and since all clause vertices have degree one in $T$, tree $T$ necessarily contains  exactly three out of the four edges of the 4-cycle defined by
$v_i,p_i,r_i,n_i$. Since $c(1,2)=1$ and $\diam(T)=0$, the missing edge is necessarily either $\{p_i,r_i\}$ or $\{r_i,n_i\}$.  We define an assignment $\psi$ of the variables $x_1, \ldots, x_n$ as follows: for $i \in [n]$, if the edge $\{r_i,n_i\}$ belongs to $T$, we set $x_i$ to true; otherwise, we set $x_i$ to false. We claim that $\psi$ satisfies $\varphi$. Indeed, let $c_j$ be a vertex in $G$ corresponding to an arbitrary clause of $\varphi$. Since $c_j$ has degree one in $T$, it is attached to exactly one of the vertices $p_i,r_i,n_i$ for some $i \in [n]$. Suppose that the edge containing $c_j$ corresponds to a positive occurrence of $x_i$, the other case being symmetric. Then, by construction, necessarily the edge containing $c_j$ is either $\{c_j, p_i\}$ or $\{c_j, r_i\}$. In both cases, if the edge $\{p_i,r_i\}$ were in $T$, this edge together with $\{c_j, p_i\}$ or $\{c_j, r_i\}$ would incur a cost of 1 in $T$, contradicting the hypothesis that $\diam(T)=0$. Therefore, the edge $\{p_i,r_i\}$ cannot be in $T$, implying that the edge $\{r_i,n_i\}$ must be in $T$. According to the definition of the assignment $\psi$, this implies that variable $x_i$ is set to true in $\psi$, and therefore the clause corresponding to $c_j$ is satisfied by variable $x_i$. This concludes the proof.

\end{proof}


Note that in the above reduction the cost function $c$ does {\sl not} satisfy the triangle inequality at vertices $p_i$ or $n_i$ for $i \in [n]$, and recall that this is unavoidable since otherwise the problem would be polynomial~\cite{WiSt01}. It is worth mentioning that using the ideas in the proof of~\cite[Theorem 4 of the full version]{GOP+16} it can be proved that the \pb problem is also $\np$-hard on {\sl planar} graphs with $\Delta \leq 4$, $k=0$, and bounded number of colors; we omit the details here.



Finally, we present  the {\sf para}-$\np$-hardness result with parameter $\tw + \Delta$.

\begin{theorem}\label{thm:hard-tw-delta}
The \pb problem is $\np$-hard on planar graphs with $\tw \leq 3$ and $\Delta \leq 3$. In particular, it is {\sf para}-$\np$-hard parameterized by $\tw$ and $\Delta$.
\end{theorem}
\begin{proof} We present a reduction from the \textsc{Partition} problem, which is a typical example of a {\sl weakly} $\np$-complete problem~\cite{GJ79}. An instance of \textsc{Partition} is a multiset $S = \{ a_1, a_2, \ldots, a_n\}$ of $n$ positive integers, and the objective is to decide whether $S$ can be partitioned into two subsets $S_1$ and $S_2$ such that
$\sum_{x\in S_1} x= \sum_{x \in S_2} x = \frac{B}{2}$ where $B=\sum_{x \in S} x$.

Given an instance $S = \{ a_1, a_2, \ldots, a_n\}$ of \textsc{Partition}, we create an instance $(G, \chi, c)$ of \pb as follows.
The graph $G$ contains a vertex $r$, called the root, and for every integer $a_i$ where $i \in [n]$, we add to $G$ six vertices $u_i, u_i', m_i, m_i',d_i, d_i'$ and seven edges $\{u_i,u_i'\}$, $\{m_i,m_i'\}$, $\{d_i,d_i'\}$, $\{u_i,m_i\}$, $\{u_i',m_i'\}$, $\{m_i,d_i\}$, and $\{m_i',d_i'\}$. We denote by $H_i$ the subgraph induced by these six vertices and seven edges. We add the edges $\{r, u_1\}, \{r, d_1\}$ and, for $ i \in [n-1]$, we add the edges $\{u_i', u_{i+1}\}$ and $\{d_i', d_{i+1}\}$.
Let $G'$ be the graph constructed so far.
We then define $G$ to be the graph obtained from two disjoint copies of $G'$ by adding an edge between both roots.
Note that $G$ is a planar graph with $\Delta(G)= 3$ and $\tw(G)= 3$. (The claimed bound on the treewidth can be easily seen by building a {\sl path} decomposition of $G$ with consecutive bags of the form $\{u_{i-1}', d_{i-1}', u_{i}, d_{i}\}, \{u_i, d_i, m_i, u_i'\},   \{ d_i, m_i, u_i', m_i'\},  \{ d_i, u_i', m_i', d_i'\}, \ldots$.)

\begin{figure}[htb]
\begin{center}
\hspace*{-.0cm}\includegraphics[width=.98\textwidth]{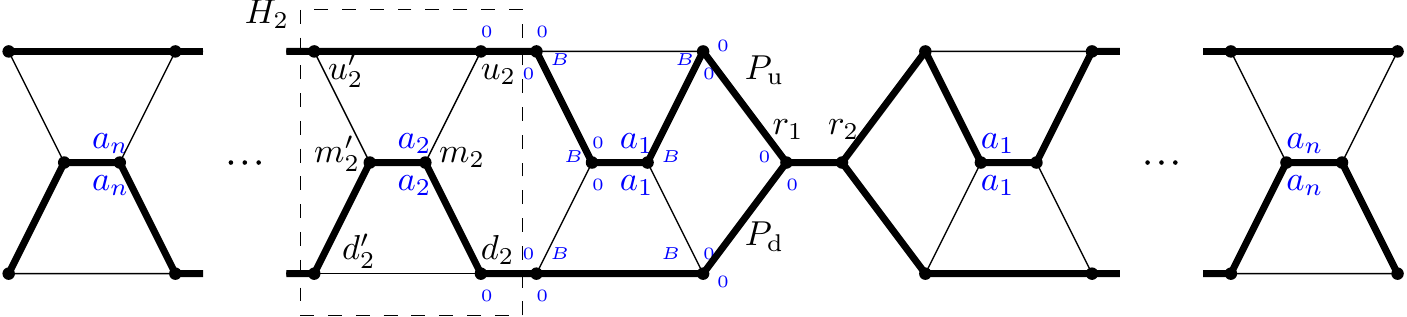}
\end{center}
\caption{Graph $G$ built in the reduction of Theorem~\ref{thm:hard-tw-delta}, where the reload costs are depicted in blue at the angle between the two corresponding edges. For better visibility, not all costs and vertex labels are depicted. The typical shape of a solution spanning tree is highlighted with thicker edges.}
\label{fig:partition}
\end{figure}

Let us now define the coloring $\chi$ and the cost function $c$. Again, for simplicity, we associate a distinct color with
each edge of $G$,
and thus it is enough to describe the
cost function $c$ for every pair of incident edges of $G$.
We define the costs for one of the copies of $G'$, and the same costs apply to the other copy. For every edge $e$ being either $\{u_i',u_{i+1}\}$ or $\{d_i',d_{i+1}\}$, for $1\leq i \leq n-1$, we set $c(e,e')=0$ for each of the four edges $e'$ incident with $e$. For every edge
$e=\{m_i, m_i'\}$, for $1\leq i \leq n$, we set $c(\{u_i, m_i\}, e) = c(\{d_i, m_i\}, e) = a_i$ and $c(e, \{m_i', u_i'\}) = c(e, \{m_i', d_i'\}) = 0$. All costs associated with the two edges containing $r$  in one of the copies $G'$ are set to $0$. For $e = \{r_1,r_2\}$, where $r_1$ and $r_2$ are the roots of the two copies of $G'$, we set $c(e,e')=0$ for each of the four edges $e'$ incident to $e$. The cost associated with any other pair of edges of $G$ is equal to $B + 1$; see Figure~\ref{fig:partition} for an illustration, where (some of) the reload costs are depicted in blue, and a typical solution spanning tree of $G$ is drawn with thicker edges.

%

We claim that the instance $S$ of \textsc{Partition} is a \textsc{Yes}-instance if and only if $G$ has a spanning tree with diameter at most $B$.

Assume first that $S$ is a \textsc{Yes}-instance of \textsc{Partition}, and let $S_1, S_2 \subseteq S$ be a solution. We define a spanning tree $T$ of $G$ with diameter $B$ as follows. We describe the subtree of $T$ restricted to one of the copies of $G'$, say $T'$. The spanning tree $T$ of $G$ is defined by union of two symmetric copies of $T'$, one in each copy of $G'$, together with the edge $\{r_1, r_2\}$. Tree $T'$ consists of the two edges $\{r, u_1\}, \{r, d_1\}$ and two paths $P_{\text{u}}, P_{\text{d}}$
(corresponding to the upper and the lower path, respectively defined as follows; see Figure~\ref{fig:partition}). For $i \in [n-1]$, the path $P_{\text{u}}$ (resp. $P_{\text{d}}$) contains the edge $\{u_i', u_{i+1}\}$ (resp. $\{d_i', d_{i+1}\}$), and if $a_i \in S_1$ we add the three edges $\{u_i,m_i\}, \{m_i, m_i'\}, \{m_i', u_i'\}$ to $P_{\text{u}}$, and the edge $\{d_i, d_i'\}$ to $P_{\text{d}}$. Otherwise, if $a_i \in S_2$, we add the edge $\{u_i, u_i'\}$ to $P_{\text{u}}$ and the three edges $\{d_i,m_i\}, \{m_i, m_i'\}, \{m_i', d_i'\}$ to $P_{\text{d}}$. Since
$\sum_{x \in S_1} x= \sum_{x \in S_2} x = \frac{B}{2}$,
it can be easily checked that both paths $P_{\text{u}}$ and $P_{\text{d}}$ have diameter $\frac{B}{2}$ in each of the two copies of $G'$, and therefore $T$ is a spanning tree of $G$ with diameter $B$.

Conversely, let $T$ be a spanning tree of $G$ with $\diam(T) \leq B$. Let $G_1,G_2$ be the two copies of $G'$ in $G$, and let $r_1,r_2$ be their respective roots. Since the edge $\{r_1, r_2\}$ is a bridge of $G$, it necessarily belongs to $T$.  By the construction of $G$, the choice of the reload costs, and since $\diam(T) \leq B -1$, it can be verified that, for $j \in \{1,2\}$, $T \cap G_j$ consists of two paths $P^j_{\text{u}}, P^j_{\text{d}}$ intersecting at the root $r_i$. Furthermore, $P^j_{\text{u}}$ (resp. $P^j_{\text{d}}$) contains the edge $\{u_i', u_{i+1}\}$ (resp. $\{d_i', d_{i+1}\}$) of the corresponding copy of $G'$, and the intersection of $P^j_{\text{u}}$ (resp. $P^j_{\text{d}}$) with the subgraph $H_i$ in the corresponding copy of $G'$ is given by either the three edges $\{u_i,m_i\}, \{m_i, m_i'\}, \{m_i', u_i'\}$ (resp. $\{d_i,m_i\}, \{m_i, m_i'\}, \{m_i', d_i'\}$) or by the edge $\{u_i, u_i'\}$ (resp. $\{d_i, d_i'\}$). Therefore,  for $j \in \{1,2\}$ and $x \in \{\text{u},\text{d}\}$, it holds that $d_x^j := \diam(P^j_{x}) =\sum_{i \in I^j_x} a_i$, where $I^j_x$ is the set of indices $i \in \{1, \ldots,n \}$ such that the edge $\{m_i, m_i'\}$ belongs to path $P^j_{x}$. Note also that, for $j \in \{1,2\}$, by construction we have that $d^j_{\text{u}} + d^j_{\text{d}} = \sum_{i=1}^n a_i$, implying in particular that
$\max\{d^j_{\text{u}},d^j_{\text{d}}\} \geq \frac{B}{2}$. On the other hand, by the structure of $T$ it holds that
\begin{equation}
\label{eq:paths} B \geq \diam(T)  \geq \max\{d^1_{\text{u}},d^1_{\text{d}}\} + \max\{d^2_{\text{u}},d^2_{\text{d}}\} \geq \frac{B}{2}  + \frac{B}{2} = B.
\end{equation}
Equation~(\ref{eq:paths}) implies, in particular, that $d^1_{\text{u}} = d^1_{\text{d}} = \frac{B}{2}$. In other words, $\sum_{i \in I^1_{\text{u}}} a_i = \sum_{i \in I^1_{\text{d}}} a_i = \frac{B}{2}$, thus the sets $I^1_{\text{u}}, I^1_{\text{d}}$ define a solution of \textsc{Partition}. This completes the proof. \end{proof}



\section{A polynomial-time algorithm on cactus graphs}
\label{sec:cactus}

In this section we present a polynomial-time algorithm to solve the \pb problem on cactus graphs, equivalently called \emph{cacti}. We first need some definitions.

A \emph{biconnected component}, or \emph{block}, of a graph is a maximal biconnected induced subgraph of it. The \emph{block tree} of a graph $G$ is a tree $T$ whose nodes are the cut vertices and the blocks of $G$. Every cut vertex is adjacent in $T$ to all the blocks that contain it. Two blocks share at most one vertex. The block tree of a graph is unique and can be computed in polynomial time~\cite{Die05}. A graph is a \emph{cactus} graph if every block of it is either a cycle or a single edge. We term these blocks as \emph{cycle block} and \emph{edge block}, respectively. It is well-known that cacti have treewidth at most 2. Given a forest $F$ and two vertices $x$ and $y$, we define $\cost_F(x,y)$ to be $\dist_{T}^{c}(x,y)$ if $x$ and $y$ are in the same tree $T$ of $F$ and where $c$ is the given reload cost function, and $\bot$ otherwise. Given a tree $T$ and a vertex $v \in V(T)$, we define the \emph{eccentricity} of $v$ in $T$ to be $\max_{v' \in V(T)}\cost_T(v,v')$.

We present a polynomial-time algorithm that solves the decision version of the problem, which we call {\sc Diameter-Tree*}: the input is an edge-colored graph $G$ and an integer $k$, and the objective is to decide whether the input graph $G$ has a spanning tree with reload cost diameter at most $k$. The algorithm to solve {\sc Diameter-Tree*} uses dynamic programming on the block tree of the input graph.

%

As we aim at a {\sl truly} polynomial-time algorithm to solve \pb, we {\sl cannot} afford  to solve the decision version for all values of $k$. To overcome this problem, we perform a double {\sl binary search} on the possible solution values and two appropriate eccentricities, resulting (skipping many technical details) in an extra factor of $(\log {\sf opt})^2$ in the running time of the algorithm, where ${\sf opt}$ is the diameter of a minimum cost spanning tree. This yields a polynomial-time algorithm solving \pb in cactus graphs.

Roughly speaking, the algorithm first fixes an arbitrary non-cut vertex $r$ of $G$ and the block $B_r$ that contains it. Then it processes the block tree of $G$ in a bottom-up manner starting from its leaves, proceeding towards $B_r$ while maintaining partial solutions for each block.
At each step of the processing, it uses an algorithm that solves an instance of the \textsc{2-Sat} problem as a subroutine. The intuition behind the instances of \textsc{2-Sat} created by the algorithm is the following.

Suppose that we are dealing with a cycle block $B$ of the block tree of $G$ (the case of an edge block being easier). Note that any spanning tree of $G$ contains all edges of $B$ except one. Let $G_B$ be the graph processed so far (including $B$). For each  potential partial solution $\Qcal$ in $G_B$, we associate, with each edge $e$ of $B$, a variable that indicates that $e$ is the non-picked edge by the solution in $B$. Now, for any {\sl two} such variables corresponding to intersecting blocks, we add to the formula of \textsc{2-Sat} essentially two types of clauses: the first set of clauses, namely $\phi_1$,  guarantees that the non-picked edges (corresponding to the variables set to \texttt{true} in the eventual assignment) indeed define a spanning tree of $G_B$, while the second one, namely $\phi_2$,  forces this solution to have diameter and eccentricity 
not exceeding the given budget $k$. The fact the $G$ is a cactus allows to prove that these constraints containing only two variables are enough to compute an optimal solution in $G_B$. 

\begin{theorem}\label{thm:cactus}
The \pb problem can be solved in polynomial time on cacti.
\end{theorem}
\begin{proof}
%

%

We start with a few more definitions needed in the algorithm. Given a graph $G$, we denote by $\spanning(G)$ the set of spanning trees of $G$, and by $\Bcal(G)$ the set of blocks of $G$. We omit $G$ from the notation if no ambiguity arises. We assume without loss of generality that $G$ is connected. For a block $B$, we denote by $\Ccal(B)$ the set of blocks that are immediate descendants of $B$ in the block tree. With a slight abuse (since we ignore the cut vertices in the block tree), we will refer to them as the \emph{children} of $B$. The \emph{parent} of a block $B$ is the first block after $B$ on the path from $B$ to $B_r$ in the block tree.
We denote by $G_B$ the subgraph of $G$ induced by the union of all descendants $B$ (including $B$ itself).
The \emph{anchor} $a(B)$ of a block $B$ is the cut vertex separating $B$ from its parent if $B \neq B_r$, and $r$ if $B=B_r$.

Let $B$ be a cycle block, $e=\{x,y\}$, and assume, without loss of generality, that $y \not = a(B)$. Clearly, the graph $G_B - e$ is connected. Moreover, $a(B)$ is a cut vertex of $G_B - e$ unless $x=a(B)$. For $z \in \set{x,y}$ we define $S_B^{z,e}$ as the set of vertices that are reachable from $z$ in $G_B - e$ without traversing $a(B)$. See Figure~\ref{fig:cactus} for an illustration. We denote the subgraph of $G_B - e$ induced by $S_B^{z,e}$, as $G_B^{z,e}$. Note that $z$ and $a(B)$ are in  $S_B^{z,e}$ and if $x = a(B)$ then $S_B^{x,e} = \{a(B)\}$. Since the degree of $a(B)$ in $G_B - e$ is at most two, a spanning tree $T$ of $G_B - e$ is a union of two spanning trees, a tree $T[S_B^{x,e}]$ spanning $G_B^{x,e}$ and a tree $T[S_B^{y,e}]$ spanning $G_B^{y,e}$. Moreover, $T[S_B^{x,e}]$ and $T[S_B^{y,e}]$ intersect only at $a(B)$.

\begin{figure}[h]
\centering
  \begin{tikzpicture}
\node at (2.2, -2) {$b$};
\node at (3, 0.3) {$r$};

\node at (1.7, -3) {$x$};
\node at (3.3, -3) {$y$};

\node at (1.5, -1) {$a(B)$};
\node at (2.5, -2.8) {$e$};
\node at (1, -5) {$S_B^{x,e}$};
    \def\a{2.5};
    \filldraw (1,0) circle (\a pt);
    \filldraw (2,0) circle (\a pt);
    \filldraw (3,0) circle (\a pt);
    \filldraw (2,-1) circle (\a pt);
    \filldraw (3,-1) circle (\a pt);
    \filldraw (4,-0.5) circle (\a pt);
    \filldraw (4,0.5) circle (\a pt);
    \filldraw (5,-0.5) circle (\a pt);
    \filldraw (5,0.5) circle (\a pt);
    \filldraw (5,-1.5) circle (\a pt);
    \filldraw (4,-1.5) circle (\a pt);

    \filldraw (2,-1) circle (\a pt);
    \filldraw (1.5,-2) circle (\a pt);
    \filldraw (2,-3) circle (\a pt);
    \filldraw (3,-3) circle (\a pt);
    \filldraw (3,-2) circle (\a pt);
    \filldraw (1,-2) circle (\a pt);
    \filldraw (0,-2) circle (\a pt);
    \filldraw (0,-3) circle (\a pt);
    \filldraw (1,-3) circle (\a pt);
    \filldraw (2,-4) circle (\a pt);

    \draw (1,0) -- (2,0) -- (3,0) -- (4,-0.5) -- (3,-1) -- (2,-1) -- (2,0);
    \draw (4,-.5) -- (4, 0.5)  -- (5, 0.5) -- (5, -0.5) -- (4,-0.5);
    \draw (4,-0.5) -- (4, -1.5) -- (5, -1.5) -- (4, -0.5);

    \draw (2, -1) -- (1.5, -2) -- (2, -3) -- (3, -3) -- (3, -2) -- (2, -1);
\draw (1.5, -2) -- (1, -2) -- (0, -2) -- (0, -3) -- ( 1, -3) -- (1, -2);
\draw (2, -4) -- (2, -3);
\draw[dotted] (-0.5,-0.5) -- +(2.8, 0) -- +(2.8, -4) -- +(0,-4) -- +(0,0);

\draw[line width = 1.5pt] (2,-1) -- (1.5, -2) -- (2,-3);
\draw[line width = 1.5pt] (2,-3) -- (2,-4);
\draw[line width = 1.5pt](1.5, -2) -- (0,-2) -- (0,-3);
\draw[line width = 1.5pt](1,-2) -- (1,-3);

  \end{tikzpicture}

\caption{A cactus with $8$ blocks, $5$ cycle blocks, and $3$ edge blocks. The vertices inside the dotted rectangle are the vertices of $S_B^{x,e}$ and the bold path corresponds to a possible $R_B^{x,e}$.}
\label{fig:cactus}
\end{figure}

We proceed with the description of the algorithm. At every block $B$, we compute a function $\lambda_B : E(B) \rightarrow \spanning(G_B) \cup \{\bot\}$ of partial solutions.
If $B$ is an edge block consisting of the edge $e$, then $\lambda_B(e)$ is:
\begin{itemize}
\item a spanning tree of $G_B$,
\item of diameter at most $k$
\item that minimizes the eccentricity of $a(B)$,
\end{itemize}
\noindent
if such a tree exists, and $\bot$ otherwise.

If $B$ is a cycle block and $e=\{x,y\}$ an edge of $B$, then $\lambda_B(e)$ is:
\begin{itemize}
\item a spanning tree $T$ of $G_B - e$,
\item of diameter at most $k$
\item that minimizes the eccentricities of $a(B)$ in both $T[S_B^{x,e}]$ and $T[S_B^{y,e}]$
\end{itemize}
\noindent
if such a tree exists, and $\bot$ otherwise.
Note that, as $G_B^{x,e}$ and $G_B^{y,e}$ have only the vertex $a(B)$ in common,
minimizing the eccentricities of $a(B)$ in $T[S_B^{x,e}]$ and
minimizing the eccentricities of $a(B)$ in  $T[S_B^{y,e}]$ are two independent objectives.

If for some block $B$ we have $\lambda_B(e)=\bot$ for every edge $e$ of $B$, then $G_B$ (and therefore $G$ as well) does not contain a spanning tree of diameter at most $k$. In this case the algorithm stops and returns \textsc{No}. Otherwise, the processing continues until finally $B_r$ is processed successfully and the algorithm returns \textsc{Yes}, since there exists $e \in E(B_r)$ such that $\lambda_{B_r}(e) \not = \bot$ which constitutes a spanning tree of $G$ with diameter at most $k$.

Given a cycle block $B$, an edge $e= \{x,y\}$ of $B$, a  subgraph $T$
of $G_B$, and two integers $i$ and $j$, we say that $T$ satisfies the \emph{$(e,i,j)$-condition} if:
\begin{itemize}
\item $T$ is a tree, of diameter at most $k$, that does not contain $e$,
\item the eccentricity of $a(B)$ in $T[S_B^{x,e}]$ is at most $i$, and
\item the eccentricity of $a(B)$ in $T[S_B^{y,e}]$ is at most $j$.
\end{itemize}

Given an edge block $B$, a  subgraph $T$ of $G_B$, and an integer $i$, we say that $T$ satisfies the
\emph{$(i)$-condition} if:
\begin{itemize}
\item $T$ is a tree of diameter at most $k$ and
\item the eccentricity of $a(B)$ in $T$ is at most $i$.
\end{itemize}

Let us fix a block $B$, and an edge $e$ of $E(B)$. In the sequel our goal is to describe how to compute $\lambda_B(e)$.
We can assume that for every child $C$ of $B$, the function $\lambda_C$ has already been computed and $C$ contains at least one edge $e'$ such that $\lambda_C(e') \not = \bot$, since otherwise the algorithm would have stopped.

We define $T^e$ to be the tree obtained by taking the union of all the following:
\begin{itemize}
\item the graph $B - e$  if $B$ is a cycle block,
\item the graph $B$ if $B$ is an edge block, and
\item $\lambda_C(e_C)$ for every child $C$ of $B$ that is an edge block containing only the edge $e_C$.
\end{itemize}

For every child $C$ of $B$ that is a cycle block, for every edge $e'$ of $C$ such that $\lambda_C(e') \not = \bot$ and for $x' \in e'$ , the tree $R_C^{x',e'}$ is
$\lambda_C(e')[S_B^{x',e'}]$.
Note that, given a child $C$ of $B$ that is a cycle block, and three vertices $v,v',v''$ of $V(C)$ such that $v \not = v''$, $v' \not = a(C)$, and $\{v,v'\}$ and $\{v',v''\}$ are in $E(C)$,  if  $R_C^{v,\{v,v'\}}$ and $R_C^{v',\{v',v''\}}$ are defined, then $R_C^{v,\{v,v'\}}$ is a subgraph of $R_C^{v',\{v',v''\}}$.
We define 
$\Rcal_B = \{R_C^{x',e'} \mid C \in \Ccal(B) \mbox{ is a cycle block}, e' \in E(C),\lambda_C(e') \not = \bot,  x' \in e' \}$.

For $\Qcal \subseteq \Rcal_B$ we denote by $T_\Qcal^e$ the graph obtained by taking the union of $T^e$ and  $\Qcal$.
If there exists $R \in \Rcal_B$ such that $\Qcal=\{ R \}$, we write $T_R^e$ instead of $T_\Qcal^e$.
We define $\texttt{close}_{\Rcal_B}(\Qcal)$ to be the set of  elements of $\Rcal_B$ that are subgraphs of $T_\Qcal^e$.
Note that $T_\Qcal^e = T_{\texttt{close}_{\Rcal_B}(\Qcal)}^e$.

If $B$ is a cycle block, we define for each $i,j \in \intv{0,k}$ the set $\Rcal_B^{(e,i,j)} = \{R \in \Rcal_B\mid \mbox{$T_R^e$ satisfies the  $(e,i,j)$-condition} \}$.
If $B$ is an edge block, we define for each $i \in \intv{0,k}$ the set $\Rcal_B^{(i)} = \{R \in \Rcal_B\mid \mbox{$T_R^e$ satisfies the  $(i)$-condition} \}$.

Note that, if $B$ is a cycle block (resp. an edge block), then for each $i,j \in \intv{0,k}$ and for each $R_1, R_2 \in \Rcal_B$ such that $R_2$ is a subtree of $R_1$, then if $R_2 \not \in \Rcal_B^{(e,i,j)}$ (resp. $R_2 \not \in \Rcal_B^{(i)}$) , we have $R_1 \not \in \Rcal_B^{(e,i,j)}$ (resp. $R_1 \not \in \Rcal_B^{(i)}$).

We associate a boolean variable ${\bf v}(R) = {\bf v}_C^{x',e'}$ with each $R=R_C^{x',e'} \in \Rcal_B$.
With a slight abuse of notation, we say that a set $\Qcal \subseteq \Rcal_B$ satisfies a formula $\phi$ over these variables if  $\phi$ is satisfied when each variable of
$\{{\bf v}(R) \mid R \in \Qcal\}$ is set to true and each variable of
$\{{\bf v}(R) \mid R \in \Rcal_B \sm \Qcal\}$ is set to false simultaneously.
In the following we are going to build three formulas $\phi_0$, $\phi_1$, and $\phi_2$, and if $\Qcal \subseteq \Rcal_B$ satisfies $\phi_0 \wedge\phi_1 \wedge \phi_2$ then this implies that $T_\Qcal^e$ is a correct value for $\lambda_B(e)$.

Along with the description, the reader is referred to Figure~\ref{fig:phi} to get some intuition about the formulas $\phi_0$, $\phi_1$, and $\phi_2$. In this figure, we have a cycle block $B$ with two children $C_1$ and $C_2$. As we will see later, when computing $\lambda_B(e)$ in this example, we have:
    \begin{eqnarray*}
      \phi_0 & = & ({\bf v}(R_{C_1}^{v',\{v',v''\}}) \Rightarrow {\bf v}(R_{C_1}^{v,\{v,v'\}})) \wedge ({\bf v}(R_{C_1}^{v'',\{v'',v\}}) \Rightarrow {\bf v}(R_{C_1}^{v',\{v',v''\}})) \wedge\\
 &  & ({\bf v}(R_{C_1}^{v'',\{v'',v'\}}) \Rightarrow {\bf v}(R_{C_1}^{v,\{v,v''\}})) \wedge ({\bf v}(R_{C_1}^{v',\{v',v\}}) \Rightarrow {\bf v}(R_{C_1}^{v'',\{v'',v''\}})) \wedge\\
 &  & ({\bf v}(R_{C_2}^{w',\{w',w''\}}) \Rightarrow {\bf v}(R_{C_2}^{w,\{w,w'\}})) \wedge ({\bf v}(R_{C_2}^{w'',\{w'',w\}}) \Rightarrow {\bf v}(R_{C_2}^{w',\{w',w''\}})) \wedge\\
 &  & ({\bf v}(R_{C_2}^{w'',\{w'',w'\}}) \Rightarrow {\bf v}(R_{C_2}^{w,\{w,w''\}})) \wedge ({\bf v}(R_{C_2}^{w',\{w',w\}}) \Rightarrow {\bf v}(R_{C_2}^{w'',\{w'',w''\}})), \\
    \end{eqnarray*}
    \begin{eqnarray*}
      \phi_1 & = & ({\bf v}(R_{C_1}^{v',\{v',v''\}}) \vee {\bf v}(R_{C_1}^{v',\{v',v\}}))\wedge (\overline{{\bf v}(R_{C_1}^{v',\{v',v''\}})} \vee \overline{{\bf v}(R_{C_1}^{v',\{v',v\}})}) \wedge \\
       & & ({\bf v}(R_{C_1}^{v'',\{v'',v'\}}) \vee {\bf v}(R_{C_1}^{v'',\{v'',v\}}))\wedge (\overline{{\bf v}(R_{C_1}^{v'',\{v'',v'\}})} \vee \overline{{\bf v}(R_{C_1}^{v'',\{v'',v\}})}) \wedge \\
       & & ({\bf w}(R_{C_2}^{w',\{w',w''\}}) \vee {\bf w}(R_{C_2}^{w',\{w',w\}}))\wedge (\overline{{\bf w}(R_{C_2}^{w',\{w',w''\}})} \vee \overline{{\bf w}(R_{C_2}^{w',\{w',w\}})}) \wedge \\
       & & ({\bf w}(R_{C_2}^{w'',\{w'',w'\}}) \vee {\bf w}(R_{C_2}^{w'',\{w'',w\}}))\wedge (\overline{{\bf w}(R_{C_2}^{w'',\{w'',w'\}})} \vee \overline{{\bf w}(R_{C_2}^{w'',\{w'',w\}})}),\mbox{ and}
    \end{eqnarray*}\vspace{.1cm}
    
for every two vertices of $V(C_1) \cup V(C_2)$, say $v'$ and $w''$, the clause $\overline{{\bf v}(R_{C_1}^{v',\{v',v''\}})} \vee \overline{{\bf v}(R_{C_2}^{w'',\{w'',w'\}})}$ is a clause of $\phi_2$ if and only if the path defined by $v',v,w,w''$ has diameter greater than $k$.
In the general case, the clauses deal with $T^e_{\{R_{C_1}^{v',\{v',v''\}},R_{C_2}^{w'',\{w'',w'\}}\}}$ instead of the path $v',v,w,w''$, but the main idea behind the clauses is the same.

\begin{figure}
  \centering
  \begin{tikzpicture}

    \node at (0,-0.5) {$B$};
    \node at (-2.5,-1.5) {$C_1$};
    \node at (2.5,-1.5) {$C_2$};
    \node at (0,+0.3) {$u = a(B)$};
    \node at (1.3,-1) {$w$};
    \node at (-1.3,-1) {$v$};
    \node at (2,-2.3) {$w'$};
    \node at (1,-2.3) {$w''$};
    \node at (-2,-2.3) {$v'$};
    \node at (-1,-2.3) {$v''$};
    \node at (0.7,-0.4) {$e$};
    \def\a{2.5};
    \filldraw (0,0) circle (\a pt);
    \filldraw (-1,-1) circle (\a pt);
    \filldraw (1,-1) circle (\a pt);
    \filldraw (-1,-2) circle (\a pt);
    \filldraw (-2,-2) circle (\a pt);
    \filldraw (1,-2) circle (\a pt);
    \filldraw (2,-2) circle (\a pt);

    \draw (0,0) -- (-2,-2) -- (-1,-2) -- (-1,-1) -- (1,-1) -- (1,-2) -- (2,-2) -- (0,0);

  \end{tikzpicture}

  \caption{Example of a cycle block $B$ with two children $C_1$ and $C_2$.}
  \label{fig:phi}
\end{figure}
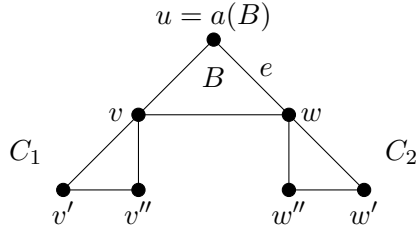


\medskip

We construct a \textsc{2-Sat} formula $\phi_0$ such that for each $R_1, R_2 \in \Rcal_B$ where $R_2$ is a subgraph of $R_1$, $\phi_0$ contains the clause ${\bf v}(R_1) \Rightarrow {\bf v}(R_2)$.
It is easy to show that given $\Qcal \subseteq \Rcal_B$, $\Qcal$ satisfies $\phi_0$ if and only if $\Qcal = \texttt{close}_{\Rcal_B}(\Qcal)$.

We construct a \textsc{2-Sat} formula $\phi_1$ as follows.
For every child $C$ of $B$ that is a cycle block, and every two consecutive edges  $e_1 = \{v_1,v_2\}$ and $e_2 = \{v_2, v_3\}$ of $C$ such that $a(C) \notin \{ v_1, v_2 \}$ and $R_C^{v_2,e_1}, R_C^{v_2,e_2} \in \Rcal_B$, we add to $\phi_1$ two clauses ${\bf v}_C^{v_2,e_1} \vee {\bf v}_C^{v_2,e_2}$ and  $\overline{{\bf v}_C^{v_2,e_1}} \vee \overline{{\bf v}_C^{v_2,e_2}}$.
With this definition of $\phi_1$, we now state the following lemma.

\begin{lemma}
\label{lemma:phi1}
Let $\Qcal$ be a subset of $\Rcal_B$ such that $\Qcal$ satisfies $\phi_0$.
  $\Qcal$ satisfies $\phi_1$ if and only if $T_\Qcal^e$ is a spanning tree of $G_B$.
\end{lemma}
\begin{proof}
Let $\Qcal \subseteq \Rcal_B$.
First assume that $T_\Qcal^e$ is a spanning tree of $G_B$.
Let $C \in \Ccal(B)$ be a cycle block, and let $e_1 = \{v_1,v_2\}$ and $e_2 = \{v_2, v_3\}$ be two consecutive edges of $C$ such that $e_1 \not = e_2$, $v_1 \not = a(C)$, and $v_2 \not = a(C)$.
As $T_\Qcal^e$ is connected, the clause ${\bf v}_C^{v_2,e_1} \vee {\bf v}_C^{v_2,e_2}$ is satisfied.
As $T_\Qcal^e$ does not contain any cycle,  the clause  $\overline{{\bf v}_C^{v_2,e_1}} \vee \overline{{\bf v}_C^{v_2,e_2}}$ is satisfied.

Assume now that $\Qcal$ satisfies $\phi_1$, and let $z$ be a vertex of $G_B$. If $z \in V(B)$, then there is a path from $z$ to $a(B)$ in $T^e$ and hence also in $T_\Qcal^e$.
Otherwise, let $C_z \in \Ccal(B)$ be the block such that $z \in V(G_{C_z})$, and let $s(z)$ be the only vertex of $V({C_z})$ such that $s(z) \not = a(C_z)$ and each path in $G_B$ from $a(B)$ to $z$ contains $s(z)$.
Note that if $z \in V(C_z)$, then $s(z) = z$.
If $C_z$ is an edge block, then $z \in V(T^e)$; therefore, there is a path from $a(B)$ to $z$ in $T_\Qcal^e$.
Otherwise, if $C_z$ is a cycle block, then the condition ${\bf v}_C^{v_2,e_1} \vee {\bf v}_C^{v_2,e_2}$, with $v_2 = s(z)$, ensures that $z \in T_\Qcal^e$ and that there is a path in $T_\Qcal^e$ from $a(B)$ to $z$. Thus, $T_\Qcal^e$ is connected and $V(T_\Qcal^e) = V(G_B)$.
We need to show that $T_\Qcal^e$ does not contain any cycle.
By construction of $T_\Qcal^e$, if it contains a cycle, this cycle should be $C$ where $C \in \Ccal(B)$.
The condition $\overline{{\bf v}_C^{v_2,e_1}} \vee \overline{{\bf v}_C^{v_2,e_2}}$ ensures that $C$ is not a subgraph of $T_\Qcal^e$.
\end{proof}

We build a formula $\phi_2$ over the variables $\{{\bf v}(R)\mid R \in \Rcal_B^{(e,k,k)}\}$.
For each $R_1, R_2 \in  \Rcal_B^{(e,k,k)}$, $R_1 \not = R_2$, if $T^e_{\{R_1,R_2\}}$ has diameter greater than $k$, then we add the clause $\overline{{\bf v}(R_1)} \vee \overline{{\bf v}(R_2)}$ to $\phi_2$.
With this definition of $\phi_2$, we now state the  following lemma.
\begin{lemma}
\label{lemma:phi12}
Let $B$ be a cycle block (resp. an edge block), $i$ and $j$ be two integers of $\intv{0,k}$,
and $\Qcal$ be a subset of $\Rcal_B^{(e,i,j)}$ (resp. of $\Rcal_B^{(i)}$) such that $\Qcal$ satisfies $\phi_0$ and $\phi_1$.
$\Qcal$ satisfies $\phi_2$ and $T^e$ satisfies the $(e,i,j)$-condition (resp. the $(i)$-condition)
if and only if $T_\Qcal^e$ is a spanning tree of $G_B$
that satisfies the $(e,i,j)$-condition (resp. the $(i)$-condition).
\end{lemma}
\begin{proof}
Assume that $B$ is a cycle block.
Let $i,j$ be two integers in $\intv{0,k}$ and let $\Qcal \subseteq \Rcal_B^{(e,i,j)}$.

First assume that $T_\Qcal^e$ is a spanning tree of $G_B$ that satisfies the $(e,i,j)$-condition.
This directly implies that $T^e$ also satisfies the $(e,i,j)$-condition.
It remains to show  that $\Qcal$ satisfies $\phi_2$.
For this, assume that there exist $R_1$ and $R_2$ in $\Qcal$ such that $T_{\{R_1,R_2\}}^e$ has diameter more than $k$. Since $T_{\{R_1,R_2\}}^e$ is a subtree of $T_\Qcal^e$, this implies that $T_\Qcal^e$ also has diameter more than $k$, which is a contradiction because $T_\Qcal^e$ satisfies the $(e,i,j)$-condition.

Assume now that $\Qcal$ satisfies  $\phi_2$ and $T^e$ satisfies the $(e,i,j)$-condition.
As $\Qcal$ satisfies $\phi_1$, we know by Lemma~\ref{lemma:phi1} that $T_\Qcal^e$ is a spanning tree of $G_B$.
As $\Qcal \subseteq \Rcal_B^{(e,i,j)}$ and $T^e$ satisfies the $(e,i,j)$-condition, then
the eccentricity of $a(B)$ in $T_\Qcal^e[S_B^{x,e}]$ is at most $i$, and
the eccentricity of $a(B)$ in $T_\Qcal^e[S_B^{y,e}]$ is at most $j$.
Indeed, let $z \in S_B^{x,e}$.
If $z \in V(B)$ then, as $T^e$ satisfies the $(e,i,j)$-condition, we have that  $\texttt{cost}_{T_\Qcal^e}(a(b),z) \leq i$
If $z \not \in V(B)$ then, as $T_\Qcal^e$ is a spanning tree of $G_B$, we have that there exists $R \in \Qcal$ such that $z \in V(R)$.
By definition of $\Rcal_B^{(e,i,j)}$, we obtain that $\texttt{cost}_{T_\Qcal^e}(a(b),z) \leq i$
The same argument applies if $z \in S_B^{y,e}$.
It remains to show that $T_\Qcal^e$ is of diameter at most $k$.
Let $z$ and $z'$ be two vertices of $T_\Qcal^e$.
If both $z$ and $z'$ are in $V(B)$, then as $T^e$ satisfies the $(e,i,j)$-condition, this implies that  $\texttt{cost}_{T_\Qcal^e}(z,z') \leq k$.
If $z \in V(B)$ and $z' \not \in V(B)$ then, as $T_\Qcal^e$ is a spanning tree of $G_B$, there exists $R' \in \Qcal$ such that $z' \in V(R')$.
As $R' \in \Rcal_B^{(e,i,j)}$, $\texttt{cost}_{T_\Qcal^e}(z,z') \leq k$.
Otherwise, if both $z$ and $z'$ are not in $V(B)$,  since $T_\Qcal^e$ is a spanning tree of $G_B$, there exist $R,R' \in \Qcal$ such that $z \in V(R)$ and $z' \in V(R')$.
If $R = R'$, then as $R' \in \Rcal_B^{(e,i,j)}$,  $\texttt{cost}_{T_\Qcal^e}(z,z') \leq k$.
Otherwise, as $\Qcal$ satisfies $\phi_2$, then $T_{\{R,R'\}}^e$ has diameter at most $k$; therefore, $\texttt{cost}_{T_\Qcal^e}(z,z') \leq k$.

The same arguments apply if $B$ is an edge block.
\end{proof}


\begin{lemma}
\label{lemma:cor}
If $B$ is a cycle block (resp. an edge block) and
if there exists a spanning tree $\hat{T}_B$ of that satisfies the $(e,i,j)$-condition (resp. $(i)$-condition) for some $i,j \in \intv{0,k}$, then
there exists $\Qcal \subseteq \Rcal_B^{(e,i,j)}$ (resp. $\Qcal \subseteq \Rcal_B^{(i)}$) that satisfies $\phi_0$, $\phi_1$, and $\phi_2$.
\end{lemma}
\begin{proof}

For readability, we consider the case where $B$ is an edge block.
Let $x = a(B)$.
Assume that there exists  $\hat{T}_B$, a spanning tree of $G_B$, that satisfies the $(i)$-condition for some $i \in \intv{0,k}$.
We define $\Qcal = \texttt{close}_{\Rcal_B}(\{R_C^{x',e'} \mid C \in \Ccal(B), C$ is a cycle block, $e'\in E(C), e' \not \in E(\hat{T}_B[V(C)]), x' \in e'\})$ and we claim that
$\Qcal$  satisfies $\phi_0$, $\phi_1$, and $\phi_2$.
By definition of $\texttt{close}_{\Rcal_B}$, $\Qcal$ satisfies $\phi_0$.
It is not difficult to see that $T_\Qcal^e$ is a spanning tree and so, by Lemma~\ref{lemma:phi1}, $\Qcal$ satisfies $\phi_1$.
Let $z$ be a vertex of $V(G_B) \sm V(B)$, and let $C \in \Ccal(B)$ such that $z \in V(G_C)$.
The path in $\hat{T}_B$ and the path in $T_\Qcal^e$ from $a(B)$ to $z$ use exactly the same edges of $C$.
This implies that $\texttt{cost}_{T_\Qcal^e}(a(B),z) \leq i$.
Otherwise $\hat{T}_B[V(G_C)]$ would have been a better value for $\lambda_C(e')$ for the only compatible edge $e' \in E(C)$.
With the same arguments we show that  $T_\Qcal^e$ has diameter at most $k$.
This implies that $T_\Qcal^e$ satisfies the $(i)$-condition, and so
$\Qcal \subseteq \Rcal_B^{(i)}$ and $\Qcal$ satisfies $\phi_2$.

The same arguments also work if $B$ is a cycle block but we should take care about the part that is in $S_B^{x,e}$ and the part that is in $S_B^{y,e}$ separately.
\end{proof}


We now have all the elements to compute the value $\lambda_B(e)$.
We assume that $B$ is a cycle block (resp. an edge block).
If there is no  $\Qcal \subseteq \Rcal_B^{(e,k,k)}$ (resp. $\Qcal \subseteq \Rcal_B^{(k)}$) that satisfies $\phi_0$, $\phi_1$, and $\phi_2$,
or  $T^e$ does not satisfy the $(e,k,k)$-condition (resp. $(k)$-condition),
 then we set $\lambda_B(e) = \bot$.
Otherwise, we aim at computing two integers $i_0$ and $j_0$ that are the smallest $i$ and $j$ in $\intv{0,k}$ such that there exists $\Qcal \subseteq \Rcal_B^{(e,i,j)}$ (resp. $\Qcal \subseteq \Rcal_B^{(i)}$) that satisfies $\phi_0$, $\phi_1$, and $\phi_2$ and such that $T^e$ satisfies the $(e,i,j)$-condition (resp. $(i)$-condition).
In order to compute $i_0$ and $j_0$, we first fix $j$ to be $k$ and do a binary search on $i$, between $0$ and $k$, to find the smallest value $i_0$ such that  there exists $\Qcal \subseteq \Rcal_B^{(e,i_0,k)}$ (resp. $\Qcal \subseteq \Rcal_B^{(i_0)}$) that satisfies $\phi_0$, $\phi_1$, and $\phi_2$ and such that $T^e$ satisfies the $(e,i_0,k)$-condition (resp. $(i_0)$-condition).
We fix this value of $i_0$  and  we do a second  binary search, this time on $j$, between $0$ and $k$, to find the smallest value $j_0$ such that  there exists $\Qcal \subseteq \Rcal_B^{(e,i_0,j_0)}$ (resp. $\Qcal \subseteq \Rcal_B^{(i_0)}$) that satisfies $\phi_0$, $\phi_1$, and $\phi_2$ and such that $T^e$ satisfies the $(e,i_0,j_0)$-condition (resp. $(i_0)$-condition).
We fix this value of $j_0$ and we also fix  $\Qcal \subseteq \Rcal_B^{(e,i_0,j_0)}$ (resp. $\Qcal \subseteq \Rcal_B^{(i_0)}$) that satisfies $\phi_0$, $\phi_1$, and $\phi_2$.
We set $\lambda_B(e) = T^e_\Qcal$.
Using Lemma~\ref{lemma:phi1} and Lemma~\ref{lemma:phi12}, we know that the graph $\lambda_B(e)$ is a spanning tree of $G_B$ that satisfies the $(e,i_0,j_0)$-condition (resp. $(i_0)$-condition). Moreover, using Lemma~\ref{lemma:cor}, we know that there is no spanning subtree of $G_B$ that satisfies the $(e,i_1,j_1)$-condition (resp. $(i_1)$-condition) with $i_1 < i_0$ or $j_1 < j_0$ (resp. $i_1 < i_0$).
This finishes the description of the algorithm.

Let us now discuss about the running time of the algorithm.
At each step, given a cycle block (resp. an edge block) $B$  and $e \in E(B)$, for each $i,j \in \intv{0,k}$, we can check if $T^e$ satisfies the $(e,i,j)$-condition (resp. the $(i)$-condition) in time $\Ocal(n^2)$.
Moreover, the number of elements in $\Rcal_B$ is linear in $n$ and for each $i,j \in \intv{0,k}$, $\Rcal_B^{(e,i,j)}$ can be computed in time $\Ocal(n^2)$.
As $\Rcal_B$ contains at most $\Ocal(n)$ elements, then the  \textsc{2-Sat} formulas $\phi_0$, $\phi_1$, and $\phi_2$ contain at most $\Ocal(n^2)$ clauses.
We can check for each of the $\Ocal(n^2)$ possible clauses if it is in $\phi_0$, $\phi_1$, or $\phi_2$ in time $\Ocal(n)$.
Hence, we can compute $\phi_0$, $\phi_1$ and $\phi_2$ in time $\Ocal(n^3)$.
As they contain at most $\Ocal(n^2)$ clauses, we can solve them in time $\Ocal(n^2)$.
Since for each block $B$ and each edge $e \in E(B)$, we perform at most two (independent) binary searches to find $i_0$ and $j_0$, we can compute $\lambda_B(e)$ in time $\Ocal(n^3 \cdot \log k)$.
Because there is a linear number of values $\lambda_B(e)$ to compute, we obtain an algorithm that solves {\sc Diameter-Tree*} in time  $\Ocal(n^4 \cdot \log k)$.
Using again a binary search on $k$ between $0$ and $2^{\lceil \log {\opt} \rceil}$, and the previous algorithm that solves {\sc Diameter-Tree*} as a subroutine, we obtain an algorithm that solves {\sc Diameter-Tree} in time $\Ocal(n^4 \cdot (\log {\opt})^2)$ where $\opt$ is the diameter of the solution.\end{proof}

\section{FPT algorithm parameterized by $k + \tw + \Delta$}
\label{sec:FTP-algo}

In this section we prove that the \pb problem is $\fpt$ on general graphs parameterized by $k$, $\tw$, and $\Delta$. The proof is based on standard, but nontrivial, dynamic programming on graphs of bounded treewidth. It should be mentioned that we can assume that a tree decomposition of the input graph $G$ of width $\Ocal(\tw)$ is given together with the input. Indeed, by using for instance the algorithm of Bodlaender \emph{et al}.~\cite{BodlaenderDDFLP16}, we can compute in time $2^{\Ocal(\tw)} \cdot n$ a tree decomposition of $G$ of width at most $5 \tw$. Note that this running time is clearly dominated by the running time stated in Theorem~\ref{thm:FPT-algo}. Recall also that, by~\cite{CyganNPPRW11,Klo94}, it is possible, given a tree decomposition to transform it in polynomial time
to a new one ${\cal D}$ of the same width and construct a collection ${\cal G}$ such that the triple $({\cal D},r,{\cal G})$ is nice.


\begin{theorem}\label{thm:FPT-algo}
The \pb problem can be solved in time
$(k^{\Delta \cdot \tw} \cdot\Delta\cdot \tw)^{\Ocal( \tw)} \cdot n^{\Ocal(1)}$.
In particular, it is $\fpt$ parameterized by $k$, $\tw$, and $\Delta$.
\end{theorem}
\begin{proof}

Before we start the description of the dynamic programming, we need some definitions.
Let $F$ be a forest and let $S$ be a set of vertices in $F$ that is good for $G$.
We define $\textsf{Reduce}(F,S)$ as the forest $F'$ that is obtained from $F$
by repetitively applying the following operations
to vertices that are not in  $N_F[S]$ as long as this is possible:
\begin{itemize}
\item removing a vertex of degree $1$ and
\item dissolving a vertex of degree $2$.
\end{itemize}


Suppose now that $\textsf{Reduce}(F,S) = F'$. We define the associated {\em reduce function} $\varphi: V(F) \to V(F') \cup E(F')$ as follows.
For every vertex $z \in V(F)$, we define $K_z$ to be the set of vertices $x$ of $V(F')$ such
that there exists a path in $F$ from $z$ to $x$ that does not use any vertex of $V(F') \sm \{x\}$.
If $K_z$ contains only one element $x$, then we define $\varphi(z) = x$, otherwise we define $\varphi(z) = K_z$. To show that $\varphi$ is well-defined, we claim that $1 \leq |K_z| \leq 2$ and if $|K_z|=2$ then $K_z \in E(F')$. Indeed, since each connected component
of $F$ contains an element of $S$, we have that $|K_z| \geq 1$.
Assume that $K_z$ contains two distinct vertices $x_1$ and $x_2$.
By definition,  we know that $x_1$ and $x_2$ are in the same connected component of $F$ and also of $F'$.
Let $P_i$ be the path from $z$ to $x_i$, $i \in \{1,2\}$, in $F$ and let $P$ be the path from $x_1$ to $x_2$ in $F[V(P_1) \cup V(P_2)]$.
By definition of $x_1$ and $x_2$, $V(P) \cap V(F') = \{x_1,x_2\}$.
Moreover, since $F$ is a forest, then $P$ is the unique path from $x_1$ to $x_2$ in $F$.
Let assume that $\{x_1,x_2\}$ is not an edge of $F'$ and let $x_3$ be a vertex of $F'$ on the path from $x_1$ to $x_2$ in $F'$.
Then $x_3$ should be in $P$.
This contradicts the fact that $V(P) \cap V(F') = \{x_1,x_2\}$.
As $F'$ is a forest, this also implies that $|K_z| \leq 2$.

We now proceed with the dynamic programming algorithm that solves {\sc Diameter-Tree*}, the decision version of {\sc Diameter-Tree}.
Let $(G,\chi,c,k)$ be an instance of {\sc Diameter-Tree*}. 
Consider a nice triple $({\cal D},r,{\cal G})$ where ${\cal D}$ is a   tree decomposition $D=(Y,{\cal X}=\{X_{t}\mid t\in V(Y)\})$ of $G$ with width at most $\tw$ and ${\cal G}=\{G_{t}\mid t\in V(Y)\}$. For each  $t\in V(Y)$ we set $w_t = |X_t|$ and $V_{t}=V(G_{t})$.
We also refer to the vertices of $X_{t}$ as {\em $t$-terminals} and to the edges that are incident to vertices in $X_{t}$ as {\em $t$-terminal edges}.
We provide  a table $\Rcal_t$ that the dynamic programming algorithm computes
for each node of ${\cal D}$. For this, we need first the notion of a \emph{$t$-pair}, that is a pair $(F,\alpha)$ where:

\begin{itemize}
\item $F$ is a forest such that
\begin{enumerate}
\item $X_{t}$ is good for $F$,
\item $X_{t} \subseteq V(F)$,
\item $N_{F}(X_{t})\subseteq N_{G}(X_{t})$,
\item $|V(F) \sm N_F[X_{t}]| \leq w_t-2$, and
\item $|\{e \in E(F) \mid e \cap X_{t} = \es\}| \leq 2 w_t-3$,
\end{enumerate}
\item $\alpha:X_{t} \times X_{t}^{F}\to \intv{0,k} \cup \{\bot\}$,
\end{itemize}
We call the vertices in $V(F) \sm N_F[X_{t}]$ {\em external} vertices of $F$
and the edges of $\{e \in E(F) \mid e \cap X_{t} = \es\}$ {\em external} edges of $F$.

We  need  the function  $\beta_{t}: {\textsf{adj}_G(X_t) \choose 2}\rightarrow \intv{0,k} \cup \{\bot\}$ so
that, for each $e_1,e_2 \in \textsf{adj}_G(X_{t})$, if there exists $x \in X_{t}$ such that $e_1 \cap e_2 = \{x\}$,
then $\beta_{t}(e_1,e_2) = c(e_1,e_2)$, otherwise $\beta_{t}(e_1,e_2) = \bot$.

Let $(F,\alpha)$ be a $t$-pair. Recall that $X_{t}^{F}$ contains all $t$-terminals
and all non-$t$-terminal edges of $F$.
Given a $t$-pair $(F,\alpha)$ as above we say that it is {\em admissible} if for every $(a,a')\in X_{t}^{F}\times X_{t}^{F}$
one of the following holds:
\begin{itemize}
\item there is no path between $a$ and $a'$ in $F$ containing a vertex in $X_{t}$,
\item one, say $a$, of $a,a''$ is a vertex in $X_{t}$ and $\alpha(a,a')\leq k$,
\item some internal vertex $b$ of the path $P$ between  $a$ and $a'$ in $F$  belongs in $X_{t}$
and $\alpha_t(b,a)+\beta_t(e^-,e^+)+\alpha_t(b,a')\leq k$, where  $e^+,e^-$ are the two edges
in $P$ that are incident to $b$.
\end{itemize}
Intuitively, the admissibility of a $t$-pair $(F,\alpha)$ assures that the transferring
cost, indicated by $\alpha$, between any two external elements
is bounded by $k$.\smallskip

It is now time to give the precise definition of the table ${\cal R}_{t}$ of our dynamic programming algorithm.
A pair $(F,\alpha)$ belongs in ${\cal R}_t$ if $G$ contains a spanning tree $\hat{T}$ 
where $\diam(T)\leq k$ and  the forest $\hat{F} = \hat{T}[V_t]$ (i.e. the restriction
of $\hat{T}$ to the part of the graph that has been processed so far) satisfies the following properties:

%
\begin{itemize}
\item $\textsf{Reduce}(\hat{F},X_t) = F$, with the reduce function $\varphi$,
\item for each $x \in X_t$ and $y \in X_t^F$,
$\alpha(x,y) = \bot$ if and only if $x$ and $y$ are in two different connected components in $F$
and if $\alpha(x,y)\not = \bot$, then for each $z \in \varphi^{-1}(y)$, $\cost_{\hat{F}}(x,z) \not = \bot$ and $\alpha(x,y)  \geq \cost_{\hat{F}}(x,z)$.
\end{itemize}\smallskip

Notice that each $(F,\alpha)$ as above is a $t$-pair. Indeed,
Conditions 1--3 follow by the fact that $\hat{T}$ is a spanning tree of $G$
and therefore $\hat{F}$ is a spanning forest of $G_{T}$.
Conditions 4 and 5 follow by the fact that the internal vertices (resp. edges) of a tree with
no vertices of degree 2 are at most two less than the number of its  leaves (resp. at most twice the number of its leaves minus three).
Moreover, the values of $\alpha$ are bounded by $k$ because  the
diameter of $\hat{T}$ is at most $k$ and therefore the same holds for all the connected
components of $\hat{F}$. Notice that, for the same reason, all pairs in ${\cal R}_{t}$ must
be admissible.

In the above definition, the external vertices and edges of $F$ correspond to the parts of $\hat{F}$
that have been ``compressed'' during the reduction operation and the function $\alpha$
stores the transfer costs between those parts and the terminals.
In this way, the trees in the $t$-pairs in ${\cal R}_{t}$ ``represent'' the restriction of all possible solutions in $G_{t}$.
Moreover, the values of $\alpha$ indicate how these partial solutions
interact with the $t$-terminals.

 Our next concern is to bound the size of ${\cal R}_{t}$.

%
\begin{claim}\label{claim:DP} For every $t\in V(Y)$, it holds that $|{\cal R}_{t}|\leq k^{\Ocal(\Delta\cdot \tw^2)}\cdot (\Delta\cdot \tw)^{\Ocal( \tw)}$.
\end{claim}
\begin{proof}
As we impose $N[X_{t}] \subseteq V(F)$, we have at most $2^{\Delta \cdot \tw}$ choices for the set $\{e \in E(F) \mid e \cap X_{t} \not = \es\}$ and at most $(\Delta \cdot \tw)^{\Ocal(\tw)}$ choices for the other edges or vertices. So the number of forest we take into consideration in $\Rcal_t$ is at most $2^{\Delta \cdot \tw}\cdot (\Delta \cdot \tw)^{\Ocal(\tw)}$.
As the number of vertices and the number of edges of $F$ is upper bounded by $\Ocal(\Delta \cdot \tw)$, the number of function $\alpha$ is at most $k^{\Ocal(\Delta\cdot \tw^2)}$.
So $|\Rcal_t| \leq k^{\Ocal(\Delta\cdot \tw^2)}\cdot (\Delta\cdot \tw)^{\Ocal( \tw)}$ and the claim holds.
\end{proof}

%
%

Clearly, $(G,\chi,c,k)$ is a \textsc{Yes}-instance if and only if  $\Rcal_{r} \not = \es$. 
We now proceed with the description of how to compute the set $\Rcal_t$ for every node $t \in \Tcal$.
For this, we will assume inductively that, for every
descendent $t'$ of $t$, the set $\Rcal_{t'}$  has already been
computed. We distinguish several cases depending on the type of node $t$:

\begin{itemize}
\item If $t$ is a {\sl leaf node}. Then $G_t = \{\es,\es\}$ and $\Rcal_t = \{((\es,\es),\varnothing)\}$. 
\item If $t$ is an {\sl  vertex-introduce   node}. Let $v$ be the insertion vertex of $X_t$
and let $t'$ be the child of $t$.
Then
\begin{eqnarray*}
R_t  & = & \big\{ \big((V(F') \cup \{v\},E(F')),{\alpha}\big) \mid \exists (F',\alpha') \in R_{t'}:\\
& & ~~~~~~~~~~\alpha = \alpha'\cup \big\{\big((v,v),0\big)\big\}\cup\big\{ \big((v,a),\bot\big)\mid  a \in X^{F'}_{t}\setminus\{v\}
\big\}.
\end{eqnarray*}
Notice that at this point $v$ is just an isolated vertex of $G_{t}$. This vertex is added in $F$ and $\alpha$
is updated with the corresponding ``void'' transfer costs.

\item If $t$ is an {\sl edge-introduce  node}.  Let $e=\{x,y\}$ be the insertion edge of $X_{t}$ and let $t'$ be the child of $t$.
We define $F''=(X_{t},\{e\})$ and we set up $\alpha'':X_{t}\times X_t^{F''}\to\intv{0,k}\cup\{\bot\}$ (notice that $X_{t}^{F''}=X_{t}$) so that $\alpha''(x,y)=\alpha''(y,x)=0$ and
is $\bot$ for all other pairs of $X_t\times X_{t}$.
Then
\begin{eqnarray*}
R_t & = &  R_{t'} \cup \{ (F,\alpha) \mid  \mbox{$(F,\alpha)$ is admissible, $F$ is a forest, and there exists   a pair } \\
 & & ~~~~~~~~~~~~~~~~~~~~~ \mbox{$(F',\alpha') \in R_{t'}$ such that\ } F = F'\cup F'' \mbox{~and~} \alpha = \alpha'\oplus_{\beta_{t}}\alpha''\}.
\end{eqnarray*}
In the above case, the single edge graph $F''$ is defined and the $F$ of each new $t$-pair is its union with $F'$.
Similarly, the function $\alpha''$ encodes the trivial transfer costs in $F''$.
Also, $\alpha$ is updated so to include the fusion of the transfer costs of $\alpha$ and $\alpha''$.

\item If $t$ is an {\sl forget node}. Let $v$ be the forget vertex and let $t'$ be the child of $t$.
Then $R_t$ contains every $t$-pair  $(F,\alpha)$ such that there exists $(F',\alpha') \in \Rcal_{t'}$
where:
\begin{itemize}

\item if $t$ is not the root of $Y$, then the connected component of $F'$ containing $v$ also contains an other element $v' \in X_{t}$ (this is necessary as $X_{t}$ should always be good for $F$),
\item $F = \textsf{Reduce}(F',X_{t}),$ with associated reduce function $\varphi$,
\item
we denote by  $Z$  the set of every edge and every vertex that is in $F'$ but not in $F$.
Moreover, if $\varphi(v)$ is a vertex, then we further set $Z\leftarrow Z\cup\{\varphi(v)\}$.
Notice also that if $z\in Z$, then $\varphi(z)=\varphi(v)$.
Then $\alpha=\alpha'|_{X_{t}\times(X_{t}^{F}\setminus \{\varphi(v)\})}\cup\big\{\big((x,\varphi(v)),\max_{y \in Z}{\alpha'(x,y)}\big)\mid x \in X_{t}\big\}$.
\end{itemize}
Notice that $F$ is further reduced because $v$ has been ``forgotten'' in $X_{t}$.
This may change the status of $v$ as follows: either $v$ is not any more in $F$ or $v$
is still in $F$ but it is  not a $t$-terminal. In the first case $\varphi(v)$ is either a vertex or an edge of $F$
and in the second
$\varphi(v)=v$. In any case  we should update the values of $\alpha(x,\phi(v))$
for every $x\in X_{t}$ to the maximum transition cost (with respect to $\alpha'$) from $x$ to some element of $Z$.



\item If $t$ is an {\sl join node}. Let $t'$ and $t''$ be the children of $t$. We define
\begin{eqnarray*}
R_t & = &  R_{t'} \cup \{ (F,\alpha) \mid  \mbox{$(F,\alpha)$ is admissible, $F$ is a forest, and there exist   two } \\
 & & ~~~~~~~~~~~~~~~~~~~~~~~~~~\mbox{pairs $(F',\alpha') \in R_{t'}$ and $(F'',\alpha'') \in R_{t''}$ such \ }\\
 & & ~~~~~~~~~~~~~~~~~~~~~~~~~~\mbox{that\ }  F = F'\cup F'' \mbox{~and~} \alpha = \alpha'\oplus_{\beta_{t}}\alpha''\}.
\end{eqnarray*}
The above case is very similar to the case of the edge-introduce  node. The only difference is that now $F''$
is now taken from ${\cal R}_{t''}$.

\end{itemize}

Taking into account Claim~\ref{claim:DP} on the bound of the size of ${\cal R}_{t}$, it is easy to verify
that, in each of the above cases, $R_t$ can be computed in $k^{\Ocal(\Delta\cdot \tw^2)}\cdot (\Delta\cdot \tw)^{\Ocal( \tw)}$ steps.
So we can solve our problem in time $k^{\Ocal(\Delta\cdot \tw^2)}\cdot (\Delta\cdot \tw)^{\Ocal( \tw)} \cdot n$, and the theorem follows.

\end{proof}

\section{Polynomially bounded costs}\label{sec:poly-costs}

So far, we have completely characterized the parameterized complexity of the \textsc{Diameter-Tree} problem for any combination of the three parameters $k$, $\tw$, and $\Delta$. In this section we focus on the special case when the maximum cost value is polynomially bounded by $n$. The following corollary is an immediate consequence of Theorem~\ref{thm:FPT-algo}.

\begin{corollary}\label{cor:XP}
 If the maximum cost value is polynomially bounded by $n$, the \pb problem is in $\xp$ parameterized by $\tw$ and $\Delta$.
\end{corollary}

From Corollary~\ref{cor:XP}, a natural question is whether the \textsc{Diameter-Tree} problem  is  $\fpt$ or ${\sf W}[1]$-hard parameterized by $\tw$ and $\Delta$, in the case where the maximum cost value is polynomially bounded by $n$. The next theorem provides an answer to this question.

\begin{theorem}\label{thm:W-hard-poly-costs}
When the maximum cost value is polynomially bounded by $n$, the \pb problem is ${\sf W}[1]$-hard parameterized by $\tw$ and $\Delta$.
\end{theorem}
\begin{proof} We present a parameterized reduction from the \textsc{Bin Packing} problem parameterized by the number of bins.  In \textsc{Bin Packing}, we are given $n$ integer item sizes $a_1, \ldots, a_n$ and an integer capacity $B$, and the objective is to partition the items into a minimum number of bins with capacity $B$. Jansen \emph{et al}.~\cite{JansenKMS13} proved that \textsc{Bin Packing} is ${\sf W}[1]$-hard parameterized by the number of bins in the solution, even when all item sizes are bounded by a polynomial of the input size. Equivalently, this version of the problem corresponds to the case where the item sizes are given in {\sl unary} encoding; this is why it is called \textsc{Unary Bin Packing} in~\cite{JansenKMS13}.

Given an instance $(\{ a_1, a_2, \ldots, a_n\}, B, k)$ of \textsc{Unary Bin Packing}, where $k$ is the number of bins in the solution and where we can assume that $k \geq 2$, we create an instance $(G, \chi, c)$ of \pb as follows. The graph $G$  contains a vertex $r$ and, for $i \in [n]$ and $j \in [k]$, we add to $G$  vertices $v_i, \ell_j^i, r_j^i$ and edges $\{r, \ell_j^1\}$, $\{v_i, \ell_j^i\}$, $\{v_i, r_j^i\}$, and $\{\ell_j^i, r_j^i\}$. Finally, for $i \in [n-1]$ and $ j \in [k]$, we add the edge $\{r_j^i, \ell_j^{i+1}\}$. Let $G'$ be the graph constructed so far; see Figure~\ref{fig:bin-packing} for an illustration.

\begin{figure}[htb]
\begin{center}
\hspace*{-.0cm}\includegraphics[width=.98\textwidth]{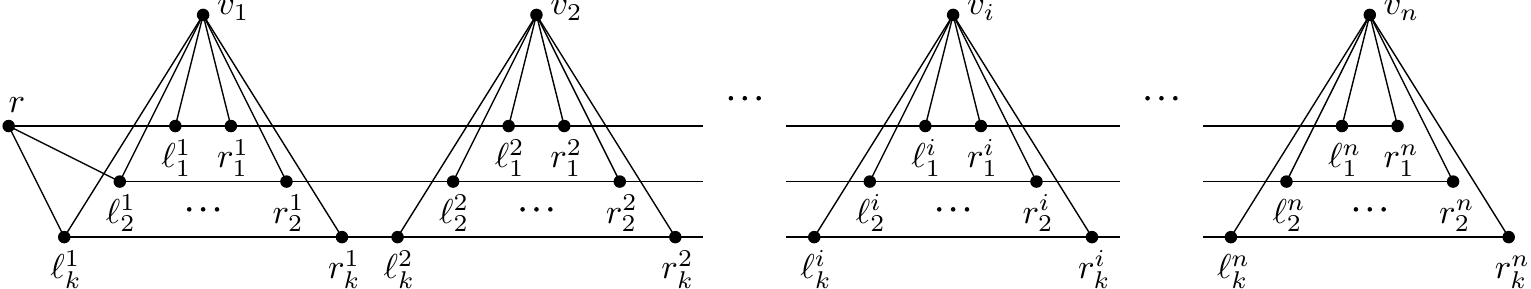}
\end{center}
\caption{Graph $G'$ built in the reduction of Theorem~\ref{thm:W-hard-poly-costs}. Reload costs are not depicted.}
\label{fig:bin-packing}\vspace{-.1cm}
\end{figure}

Similarly to the proof of Theorem~\ref{thm:hard-tw-delta}, we define $G$ to be the graph obtained by taking two disjoint copies of $G'$ and identifying vertex $r$ of both copies. Note that $G$ can be clearly built in polynomial time, and that $\tw(G) \leq k+1$ and $\Delta(G) = 2k$ (since we assume $k \geq 2)$. Therefore, $\tw(G) + \Delta(G) $ is indeed bounded by a function of $k$, as required. (Again, the claimed bound on the treewidth can be easily seen by building a {\sl path} decomposition of $G$ with consecutive bags of the form $\{v_i, \ell_1^i, \ell_2^i, \ldots, \ell_k^i, r_1^i\}, \{v_i, \ell_1^i, \ell_2^i, \ldots, \ell_{k-1}^{i}, r_1^i, r_2^i\}, \{v_i, \ell_1^i, \ell_2^i, \ldots, \ell_{k-2}^{i}, r_1^i, r_2^i, r_3^i\}, \ldots$.)

Let us now define the coloring $\chi$ and the cost function $c$. Once more, for simplicity, we associate a distinct color with each edge of $G$, and thus it is enough to describe the cost function $c$ for every pair of incident edges of $G$. The cost function is symmetric for both copies of $G'$. so we just focus on one copy. For $i \in [n]$, let $e_1, e_2$ be two distinct edges containing vertex $v_i$. We set $c(e_1,e_2) = 2B+1$ unless $e_1 = \{v_i, \ell^i_j\}$ and $e_2 = \{v_i, r^i_j\}$ for some $j \in [k]$, in which case we set $c(e_1,e_2) = a_i$. The cost associated with any other pair of edges of $G$ is set to $0$. Note that, as $(\{ a_1, a_2, \ldots, a_n\}, B, k)$ is an instance of \textsc{Unary Bin Packing}, the reload costs of the instance $(G, \chi, c)$ of \pb are polynomially bounded by $|V(G)|$.

We claim that $(\{ a_1, a_2, \ldots, a_n\}, B, k)$ is a \textsc{Yes}-instance of \textsc{Unary Bin Packing} if and only if $G$ has a spanning tree with diameter at most $2B$.

Assume first that $(\{ a_1, a_2, \ldots, a_n\}, B, k)$ is a \textsc{Yes}-instance of \textsc{Unary Bin Packing}, and let $S_1, \ldots, S_k$ be the $k$ subsets of $\{1, \ldots, n\}$ defining the $k$ bins in the solution. We define a spanning tree $T$ of $G$ with $\diam(T) \leq 2B$ as follows. For each of the two copies of $G'$, tree $T$ contains, for $i \in [n-1]$ and $j \in [k]$, edges $\{r, \ell^1_j\}$ and $\{r_j^i, \ell_j^{i+1}\}$. For $i \in [n-1]$, if the item $a_i$ belongs to the set $S_j$, we add to $T$ the two edges $\{v_i, \ell^i_j\}$ and $\{v_i, r^i_j\}$; otherwise we add to $T$ the edge $\{\ell^i_j, r^i_j\}$. Since the total item size of each bin in the solution of \textsc{Unary Bin Packing} is at most $B$, it can be easily checked that $T$ is a spanning tree of $G$ with  $\diam(T) \leq 2B$.

Conversely, let $T$ be a spanning tree of $G$ with  $\diam(T) \leq 2B$, and we proceed to define a solution $S_1, \ldots, S_k$ of \textsc{Unary Bin Packing}. Let $T_1$ and $T_2$ be the restriction of $T$ to the two copies of $G'$. By the choice of the reload costs and since $\diam(T) \leq 2B$, for every $i \in [n]$ and every $x \in \{1,2\}$, tree $T_x$ contains the two edges $\{v_i, \ell^i_j\}$ and $\{v_i, r^i_j\}$ for some $j \in [k]$, and none of the other edges incident with vertex $v_i$. Therefore, for every $x \in \{1,2\}$, tree $T_x$ consists of $k$ paths sharing vertex $r$. This implies that  $\diam(T) \geq  \frac{1}{2}\diam(T_1) + \frac{1}{2}\diam(T_2)$, and since $\diam(T) \leq 2B$, it follows that there exists $x \in \{1,2\}$ such that $\diam(T_x) \leq B$. Assume without loss of generality that $x=1$, i.e., that $\diam(T_1) \leq B$. We define the bins $S_1, \ldots, S_k$ as follows. For every $i \in [n]$, if $T_1$ contains the two edges $\{v_i, \ell^i_j\}$ and $\{v_i, r^i_j\}$, we add item $a_i$ to the bin $S_j$. Let us verify that this defines a solution of  \textsc{Unary Bin Packing}. Indeed, assume for contradiction that for some $j \in [k]$, the total item size in bin $S_j$ exceeds $B$. As bin $S_j$ corresponds to one of the $k$ paths in tree $T_1$, the diameter of this path would also exceed $B$, contradicting the fact that $\diam(T_1) \leq B$. The theorem follows. \end{proof}



\section{Concluding remarks}\label{sec:conclusion}

We provided an accurate picture of the (parameterized) complexity of the \pb problem for any combination of the parameters $k$, $\tw$, and $\Delta$, distinguishing whether the reload costs are polynomial or not. Some questions still remain open.  First of all, in the hardness result of Theorem~\ref{thm:hard-tw-delta}, we already mentioned that the bound $\Delta \leq 3$ is tight, but the bound  $\tw \leq 3$ might be improved to $\tw \leq 2$. A relevant question is whether the problem admits polynomial kernels parameterized by $k +\tw + \Delta$ (recall that it is \fpt by Theorem~\ref{thm:FPT-algo}). Theorem~\ref{thm:W-hard-poly-costs} motivates the following question: when all reload costs are bounded by a {\sl constant}, is the \pb problem $\fpt$ parameterized by $\tw+\Delta$? It also makes sense to consider the \emph{color-degree} as a parameter (cf.~\cite{FPT-by-tw-Delta}). Finally, we may consider other relevant width parameters, such as pathwidth (note that the hardness results of Theorems~\ref{thm:hard-bounded-tw},~\ref{thm:hard-tw-delta}, and~\ref{thm:W-hard-poly-costs} also hold for pathwidth), cliquewidth, treedepth, or tree-cutwidth.



\bibliographystyle{abbrv}
\bibliography{reload-FPT,Didem}

\end{document}